\documentclass[a4paper,reqno,twoside]{amsart}
\usepackage{bbm}
\usepackage[left=3.5cm,right=3.5cm,top=3.5cm,bottom=3.5cm,footskip=1.5cm,headsep=1cm,bindingoffset=0.cm,]{geometry}
\usepackage{xcolor}
\usepackage[utf8]{inputenc}
\usepackage[english]{babel}
\usepackage[T1]{fontenc} 
\usepackage{lmodern} 
\rmfamily 
\DeclareFontShape{T1}{lmr}{b}{sc}{<->ssub*cmr/bx/sc}{}
\DeclareFontShape{T1}{lmr}{bx}{sc}{<->ssub*cmr/bx/sc}{}

\numberwithin{equation}{section}
\usepackage{amsmath}
\usepackage{amssymb}
\usepackage{amsthm}
\usepackage{amsfonts}
\usepackage{mathtools}

\usepackage{mathdots}
\usepackage{caption}
\usepackage{subcaption}

\usepackage{dsfont} 
\usepackage[format=hang]{caption}
\usepackage[normalem]{ulem}

\usepackage{standalone}
\usepackage{graphicx}
\usepackage{imakeidx}
\makeindex
\usepackage{setspace}
\usepackage{appendix}
\usepackage{pdfpages}
\usepackage{upgreek}
\usepackage{tabulary}
\usepackage{booktabs}
\usepackage{extarrows}
\usepackage{tabularx}
\usepackage{float}
\restylefloat{table}
\usepackage{enumitem}
\usepackage{csquotes}
\usepackage{bm}
\usepackage{orcidlink}
\usepackage{lipsum}

\usepackage{tikz}
\usepackage{tikz-layers}
\usepackage{tikz-cd}
\usepackage[arrow, matrix, curve]{xy}
\usetikzlibrary{arrows}
\usetikzlibrary{arrows.meta}
\usetikzlibrary{matrix,positioning,decorations.pathreplacing,calc}
\usetikzlibrary{
    decorations.text,%
    decorations.markings,%
    shadows}
\usetikzlibrary{calc,intersections}
\usepackage{adjustbox}
\usepackage{graphicx}

\usepackage[
    style=numeric,
    sorting=nty,
    maxnames=99,
    maxalphanames=5,
    natbib=true,
    backend=biber,
    sortcites]{biblatex}

\DeclareNameAlias{default}{family-given}

\AtEveryBibitem{
    \clearfield{url}
    \clearfield{issn}
    \clearfield{isbn}
    \clearfield{urldate}

    \ifentrytype{book}{
        \clearfield{pages}}{%
    }
}

\usepackage{xargs}                      
\usepackage[prependcaption,textsize=tiny,textwidth=3cm]{todonotes}
\newcommandx{\unsure}[2][1=]{\todo[linecolor=red,backgroundcolor=red!25,bordercolor=red,#1]{#2}}
\newcommandx{\change}[2][1=]{\todo[linecolor=blue,backgroundcolor=blue!25,bordercolor=blue,#1]{#2}}
\newcommandx{\info}[2][1=]{\todo[linecolor=OliveGreen,backgroundcolor=OliveGreen!25,bordercolor=OliveGreen,#1]{#2}}
\newcommandx{\improvement}[2][1=]{\todo[linecolor=black,backgroundcolor=black!25,bordercolor=black,#1]{#2}}
\newcommandx{\thiswillnotshow}[2][1=]{\todo[disable,#1]{#2}}

\usepackage[noabbrev,capitalize]{cleveref}
\crefname{proposition}{Proposition}{Propositions}
\crefname{equation}{}{}

\newtheorem{theorem}{Theorem}[section]
\newtheorem{lemma}[theorem]{Lemma}
\newtheorem{proposition}[theorem]{Proposition}
\newtheorem{corollary}[theorem]{Corollary}

\theoremstyle{definition}
\newtheorem{definition}[theorem]{Definition}
\newtheorem{example}[theorem]{Example}

\newtheorem{remark}[theorem]{Remark}

\crefname{assumption}{Assumption}{Assumptions}
\crefname{definition}{Definition}{Definitions}
\crefname{corollary}{Corollary}{Corollaries}
\crefname{enumi}{item}{items}

\creflabelformat{subfigure}{#2\textsc{#1}#3}

\usepackage[final]{microtype}

\makeatletter
\newsavebox\myboxA
\newsavebox\myboxB
\newlength\mylenA

\newcommand*\xoverline[2][0.75]{%
  \sbox{\myboxA}{$\m@th#2$}%
  \setbox\myboxB\null
  \ht\myboxB=\ht\myboxA%
  \dp\myboxB=\dp\myboxA%
  \wd\myboxB=#1\wd\myboxA
  \sbox\myboxB{$\m@th\overline{\copy\myboxB}$}
  \setlength\mylenA{\the\wd\myboxA}
  \addtolength\mylenA{-\the\wd\myboxB}%
  \ifdim\wd\myboxB<\wd\myboxA%
    \rlap{\hskip 0.5\mylenA\usebox\myboxB}{\usebox\myboxA}%
  \else
    \hskip -0.5\mylenA\rlap{\usebox\myboxA}{\hskip 0.5\mylenA\usebox\myboxB}%
  \fi}
\makeatother

\usepackage{fancyhdr}

\fancypagestyle{plain}{%
    \fancyhead[C]{}
    \fancyfoot[C]{}
    
    }
\fancypagestyle{nosection}{%
    \fancyhead[CE]{}
    \fancyhead[CO]{}
    \fancyfoot[CE,CO]{\thepage}
}

\usepackage{tikz}

\usetikzlibrary{matrix} 
\usetikzlibrary{arrows,arrows.meta,bending} 

\usepackage{nicematrix}

\usetikzlibrary{hobby}

\usepackage{calc} 

\newcommand{\lz}{\ell^2(\Z)}
\newcommand{\ldz}{D^{\Z}}
\newcommand{\ldzp}{D^{\Z_{\geq 0}}}
\newcommand{\lrz}{\mathbb{L}_R(\Z)}
\newcommand{\ldm}{D^{M}}

\newcommand{\oo}[1]{\chi^{(#1)}}

\newcommand{\seq}[2]{{#1^{(#2)}}}
\newcommand{\bseq}[2]{{\bm{#1}^{(#2)}}}

\newcommand{\rp}{\mathbb{RP}^1}

\DeclareMathOperator{\prob}{P}

\DeclareMathOperator{\len}{len}

\DeclareMathOperator{\N}{\mathbb{N}}
\DeclareMathOperator{\Z}{\mathbb{Z}}

\DeclareMathOperator{\R}{\mathbb{R}}
\DeclareMathOperator{\C}{\mathbb{C}}

\DeclareMathOperator{\tr}{tr}
\renewcommand{\i}{\mathbf{i}}

\renewcommand{\qedsymbol}{$\blacksquare$}

\newcommand{\ds}{\displaystyle}

\DeclareMathOperator{\diag}{diag}
\DeclareMathOperator{\BO}{\mathcal{O}}

\renewcommand{\epsilon}{\varepsilon}
\DeclareMathOperator{\dd}{d\!}

\let\emptyset\varnothing

\renewcommand{\i}{\mathbf{i}}

\DeclareMathOperator{\iL}{{\mathsf{L}}}
\DeclareMathOperator{\iR}{{\mathsf{R}}}
\DeclareMathOperator{\iLR}{{\mathsf{L},\mathsf{R}}}

\usepackage{tcolorbox}
\usetikzlibrary{tikzmark,calc}

\newcommand{\ip}[2]{\left\langle #1, #2 \right\rangle}
\newcommand{\mc}[1]{\mathcal{#1}}
\newcommand{\on}[1]{\operatorname{#1}}

\newcommand{\abs}[1]{\left\lvert#1\right\rvert}

\newcommand{\norm}[1]{\left\lVert#1\right\rVert}

\title[Uniform Hyperbolicity in Aperiodic Resonator Systems]{Uniform Hyperbolicity, Bandgaps and Edge Modes in Aperiodic Resonator Systems}

\begin{document}
\author[H. Ammari]{Habib Ammari \,\orcidlink{0000-0001-7278-4877}}
\address{\parbox{\linewidth}{Habib Ammari\\
 ETH Z\"urich, Department of Mathematics, Rämistrasse 101, 8092 Z\"urich, Switzerland \href{http://orcid.org/0000-0001-7278-4877}{orcid.org/0000-0001-7278-4877}}.}
 \email{habib.ammari@math.ethz.ch}
 \thanks{}

\author[C. Thalhammer]{Clemens Thalhammer}
\address{\parbox{\linewidth}{Clemens Thalhammer\\
 ETH Z\"urich, Department of Mathematics, Rämistrasse 101, 8092 Z\"urich, Switzerland.}}
 \email{clemens.thalhammer@sam.math.ethz.ch}

 \author[A. Uhlmann]{Alexander Uhlmann\,\orcidlink{0009-0002-0426-6407}}
  \address{\parbox{\linewidth}{Alexander Uhlmann\\
  ETH Z\"urich, Department of Mathematics, Rämistrasse 101, 8092 Z\"urich, Switzerland \href{http://orcid.org/0009-0002-0426-6407}{orcid.org/0009-0002-0426-6407}}.}
 \email{alexander.uhlmann@sam.math.ethz.ch}

 \begin{abstract}
    We aim to characterise the spectral distributions of bi-infinite, semi-infinite, and finite aperiodic one-dimensional arrays of subwavelength resonators, constructed by sampling from a finite library of building blocks. By adopting the modern formalism of uniform hyperbolicity, we are able to strengthen and rigorously prove a Saxon-Hutner-type result, characterising the spectral gaps of the composite bi-infinite aperiodic system in terms of its constituent blocks through a strict multicone criterion. Crucial to this approach is a change of basis from transfer matrices to propagation matrices, allowing for a block-level characterisation. This approach also enables an explicit characterisation of edge-induced eigenmodes in the semi-infinite setting. Finally, we leverage finite section methods for Jacobi operators to extend our results to finite systems -- providing strict bounds for their spectra in terms of their constituent blocks.
 \end{abstract}

\maketitle

 \par

\bigskip

\\

\section{Introduction and motivation}
In this work, we study the spectral properties of aperiodically arranged one-dimensional arrays of subwavelength resonators. A subwavelength resonator is a bounded inclusion with a large contrast in material parameters compared to the background material. The high contrast enables the resonators to strongly interact with waves whose wavelength far exceeds the typical resonator size -- giving rise to subwavelength resonant eigenmodes, \emph{i.e.}, eigenmodes whose corresponding wavelength is much larger than the resonator size. The quintessential example of such resonant behaviour was first identified by Minnaert \cite{minnaert1933XVI}, who studied the exotic acoustic properties of air bubbles in water; see also \cite{ammari.fitzpatrick.ea2018Minnaert,leroy}. Subwavelength metamaterials have since emerged as a powerful paradigm for wave manipulation in both acoustic and electromagnetic settings \cite{cheben.halir.ea2018Subwavelength,cummer.christensen.ea2016Controlling,ma.sheng2016Acoustic,lemoult.kaina.ea2016Soda,lemoult.fink.ea2011Acoustic,ammari.millien.ea2017Mathematical,ammari.li.ea2023Mathematical,ammari2015superresolution}.

For our purposes, the central tool enabling the study of such subwavelength resonances is the \emph{capacitance matrix}. Namely, from \cite{feppon.cheng.ea2023Subwavelength, ammari.davies.ea2024Functional}, we know that for an array of $N$ subwavelength resonators there exist exactly $N$ subwavelength resonant frequencies, characterised in leading order by the eigenproblem for $VC$. Here, the $N\times N$ matrix $C$ is called the \emph{capacitance matrix}. It encodes the geometry of the array and is tridiagonal and symmetric in the one-dimensional setting. The matrix $V$ encodes the material parameters and is always positive-definite and diagonal.

We study aperiodic arrays of such subwavelength resonators, constructed by sampling from a library of finitely many distinct blocks (each potentially containing multiple individual resonators) and arranging them linearly. The resulting systems are strongly disordered in the sense that they are more than just a compact perturbation away from a periodic system, preventing the use of tools such as the truncated Floquet-Bloch transform \cite{ammari.barandun.ea2025Truncated} suitable for slightly aperiodic settings. As described in \cite{disorder, ammari.barandun.ea2025Universal}, such arrays display a rich variety of localisation behaviours and spectral distributions which may nonetheless be characterised by studying the individual resonator block properties. The approach of sampling from a finitely valued library of building blocks was also explored, respectively, in \cite{damanik.sims.ea2004Localization} and \cite{avila.bochi.ea2010Uniformly} for the Schrödinger operator and in the generic cocycle setting.

As the number $M$ of blocks sampled tends to infinity, the capacitance formulation strongly converges to a Jacobi operator (\emph{i.e.}, tridiagonal and self-adjoint) on the lattice $\ell^2(\Z)$. Sampling the blocks independently and identically yields an ergodic sequence of resonators and a \emph{metrically transitive} random Jacobi operator \cite{pastur1992Spectra}, which can be seen as a generalisation of the widely studied one-dimensional Schrödinger operators with random potentials \cite{aizenman2015Random, damanik.fillman2022OneDimensional}. However, the main results of this paper do not rely on the specific sampling of the block sequence but rather on a type of genericity condition called \emph{pseudo-ergodicity}, ensuring that the sequence contains every possible finite subsequence. 

The central objective of this paper is to rigorously characterise the spectral gaps (\emph{i.e.}, the regions devoid of eigenvalues) of the composite bi-infinite, semi-infinite, and finite systems in terms of their constituent blocks -- known in the field of condensed matter physics as a \emph{Saxon-Hutner-type} result. \cref{thm:saxonhutner} formulates this result in the modern language of \emph{uniformly hyperbolic} cocycles and provides an explicit sufficient condition for the existence of spectral gaps and a complementary obstruction. The example in \Cref{fig:SSC_violation} illustrates numerically the spectral consequences of violating this sufficient condition.

Uniform hyperbolicity is a powerful concept that originated from taking a dynamical system point of view to study the spectra of one-dimensional Schrödinger operators. Essentially, a random family of determinant one matrices $P(j) \in \on{SL}(2,\R)$ is \emph{uniformly hyperbolic} if they possess uniformly decaying or growing shared invariant directions. For a one-dimensional lattice Schrödinger Hamiltonian $\mc H:\ell^2(\Z)\to \ell^2(\Z)$ with the associated spectral equation $(\mc H-\lambda)\bm v = 0$, the $i$\textsuperscript{th} transfer matrix $T^\lambda(i)\in \on{SL}(2,\R)$ transfers forward the tuple $(\bseq{v}{i}, \bseq{v}{i-1})^\top$ by one entry, $(\bseq{v}{i+1}, \bseq{v}{i})^\top = T^\lambda(i)(\bseq{v}{i}, \bseq{v}{i-1})^\top$, where $\bseq{v}{i}$ denotes the $i$\textsuperscript{th} entry of $\bm v$ and the superscript $\top$ denotes the transpose. 
In a result dating back to \cite{johnson1986Exponential}, the uniform hyperbolicity of the family of transfer matrices $T^\lambda$ at some $\lambda\in \R$ was proven to be equivalent to the fact that the corresponding Schrödinger operator exhibits a gap, $\lambda\notin \sigma(\mc H)$, at that parameter, linking the transfer matrix dynamics to the spectral characteristics. 
This approach has led to significant advances in the study of random and quasiperiodic one-dimensional operators \cite{jitomirskaya1999MetalInsulator, bourgain.goldstein2000Nonperturbative, avila.bochi.ea2010Uniformly, zhang2020Uniform}.

For Jacobi operators, the appropriate generalisation for the uniform hyperbolicity of the transfer matrices turns out to be the existence of a \emph{dominated splitting} \cite{mane1978Contributions, marx2014Dominated, alkorn.zhang2022correspondence}. However, instead of working directly with the transfer matrices, it will prove very fruitful to first perform a change-of-basis to the cohomologous family of \emph{propagation matrices}. Physically, for a given resonator the propagation matrix $P^\lambda(i)$ propagates the subwavelength resonant solution $u$, in terms of its value and derivative at the left exterior edge of the resonator $(u(x_i^{\iL}), u'(x_{i}^{\iL}))^\top$, across the $i$\textsuperscript{th} resonator yielding 
the identity $$(u(x_{i+1}^{\iL}), u'(x_{i+1}^{\iL}))^\top = P^\lambda(i)(u(x_i^{\iL}), u'(x_{i}^{\iL}))^\top.$$ The key advantage of this approach is that the $i$\textsuperscript{th} propagation matrix only depends on the properties of the $i$\textsuperscript{th} resonator, making many analytical approaches significantly more tractable and natural. In particular, this will allow for the collection of individual resonator propagation matrices into block propagation matrices, enabling us to state \cref{thm:saxonhutner} in terms of block properties. Furthermore, for the semi-infinite Jacobi operator setting, this approach will allow us to explicitly characterise the additional eigenmodes due to edge effects, as stated in \cref{thm:semiinfinite-johnson}. 

Finally, we will adapt some results on the finite section method of random Jacobi operators to our setting \cite{lindner2012finite, chandler-wilde.lindner2016Coburns}. In particular, we will prove a physical version of Coburn's lemma. Consequently, we can show that the spectra of finite resonator block arrangements are fully determined by the spectra of semi-infinite arrangements. This allows us to, up to the possibility of edge modes, extend the spectral characterisation of block disordered systems in terms of their block propagation matrices to the finite case.

Our main results in \cref{thm:saxonhutner,thm:semiinfinite-johnson,thm:finite_spectra} are illustrated in \cref{fig:saxonhutner_edgemode}, showing that, up to some explicitly characterisable edge modes, the spectral distribution of the composite finite-block disordered systems is completely determined by the constituent blocks.

Based on the symmetrisation approach introduced in \cite{ammari2025competingedgebulklocalisation}, \cref{thm:saxonhutner,thm:semiinfinite-johnson,thm:finite_spectra} can be extended to provide a rigorous derivation of the limiting spectral distribution as the number of blocks tends to infinity of one-dimensional block disordered subwavelength resonator systems that are subject to non-reciprocal damping, induced by an imaginary gauge potential. This completely elucidates the mechanisms of localisation in disordered non-reciprocal systems first studied in \cite{stabilityskin,ammariMathematicalFoundationsNonHermitian2024}. 

Furthermore, in the language of block propagation matrix cocycles, a number of seminal results from the study of random Schrödinger operators may naturally be extended to our limiting bi-infinite Jacobi operator. In particular, the existence and positivity of the Lyapunov exponent \cite{furstenberg1963Noncommuting}, as well as the corresponding absence of an absolutely continuous spectrum \cite{kotani1984Ljapunov}.

\noindent
\begin{minipage}[c]{0.56\textwidth}
To clarify the hierarchy of objects used throughout the paper, we briefly describe the successive reductions underlying our approach. The geometry of an aperiodic array is first encoded by a block sequence, whose symbols specify how blocks from a finite library are concatenated; expanding each block into its constituent resonators then produces the corresponding resonator sequence. In the subwavelength regime, the continuum resonance problem is reduced, through the capacitance approximation and a symmetrisation, to a self-adjoint Jacobi operator whose coefficients are determined by the resonator parameters. Solutions of the associated spectral equation are propagated by transfer matrices; since these matrices are neither unimodular nor local at the level of a single resonator, we conjugate them to determinant-one propagation matrices that depend only on the length, spacing, and wave speed of that resonator. Multiplying the single-resonator propagation matrices within each block yields a finite family of block propagation matrices, thereby returning the analysis to the original block scale. Their action on the real projective line describes the evolution of one-dimensional solution directions, and the whole-line resolvent problem becomes equivalent, through Johnson's theorem, to uniform hyperbolicity of this finite-valued projective cocycle.

\end{minipage}%
\hfill
\begin{minipage}[c]{0.40\textwidth}
\centering
\begin{tikzpicture}[
  node distance = 7mm,
  box/.style = {draw, rounded corners=2pt, minimum width=3.1cm,
                minimum height=7mm, align=center, font=\small},
  arr/.style = {-{Stealth[length=1.8mm]}, semithick},
  lab/.style = {right, font=\scriptsize\itshape, midway, xshift=1mm}
]
  \node[box] (blk) {block sequence};
  \node[box, below=of blk] (res) {resonator sequence};
  \node[box, below=of res] (jac) {Jacobi operator};
  \node[box, below=of jac] (pro) {propagation matrices};
  \node[box, below=of pro] (bco) {block cocycle};
  \node[box, below=of bco] (prj) {projective dynamics};

  \draw[arr] (blk) -- node[lab] {expansion} (res);
  \draw[arr] (res) -- node[lab] {high contrast} (jac);
  \draw[arr] (jac) -- node[lab] {conjugation} (pro);
  \draw[arr] (pro) -- node[lab] {block products} (bco);
  \draw[arr] (bco) -- node[lab] {action on $\mathbb{RP}^1$} (prj);
\end{tikzpicture}
\end{minipage}

\medskip
For a pseudo-ergodic block sequence whose shift orbit is dense in the full block space, this can be characterised by the source--sink condition introduced below, which provides a simple explicit sufficient criterion, while the presence of an elliptic or parabolic block matrix gives a converse obstruction. The spectral consequences of violating this condition are numerically illustrated. The same projective description identifies half-line edge eigenvalues through coincidence of the boundary direction with the stable direction of the cocycle. Finally, limit-operator and finite-section arguments connect the whole-line, half-line, and finite systems, showing how their spectra decompose into bulk bands determined by the block family and possible edge-induced eigenmodes.

The paper is organised as follows. In \cref{sec:setting}, we introduce one-dimensional chains of subwavelength resonators, recall the main results on the capacitance matrix approximation, and describe the block disordered construction. In \cref{sec:jacobi}, we introduce the limiting bi-infinite Jacobi operators, as well as their transfer matrix cocycles. We then introduce the dominated splitting version of Johnson's theorem for Jacobi operators and then proceed to link the existence of a dominated splitting for the transfer matrix cocycle to the uniform hyperbolicity of the block propagation matrix cocycle via a series of lemmas. This ultimately allows us to formulate \cref{thm:saxonhutner}. The flow of reductions in \cref{sec:jacobi} is summarised in \cref{eq:saxon-hutner_schematic}. In \cref{sec:semiinf_and_finite}, our aim is to characterise the spectra of semi-infinite and finite systems. In the semi-infinite case, we find the additional modes due to edge effects by studying the block propagation matrix cocycles, culminating in \cref{thm:semiinfinite-johnson}. Finally, for finite systems, we leverage the theory of finite sections of Jacobi operators to bound the spectra of finite block disordered arrays by the appropriate semi-infinite limits, yielding \cref{thm:finite_spectra}. 

\begin{figure}
    \centering
    \includegraphics[width=0.9\linewidth]{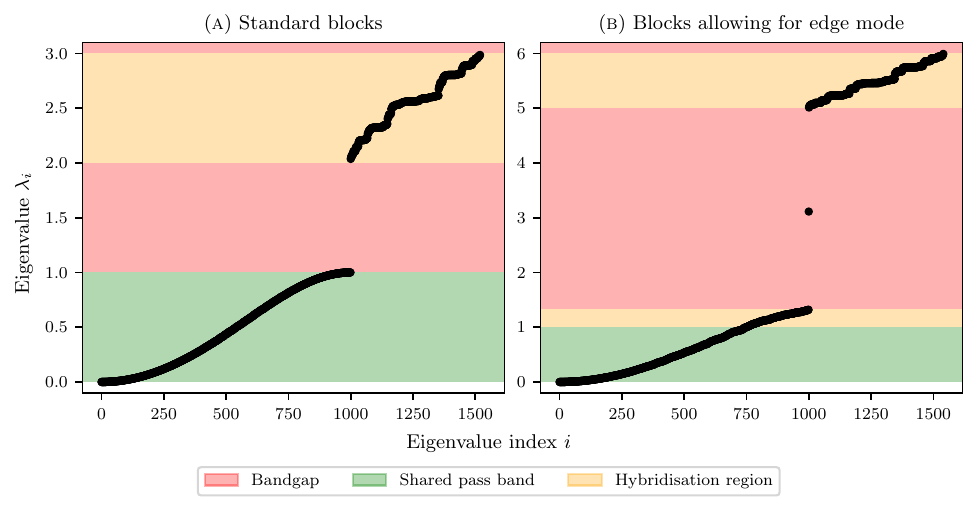}
    \caption{Illustration of \cref{thm:saxonhutner,thm:semiinfinite-johnson,thm:finite_spectra}. We plot the eigenvalues of finite block disordered systems with $M=1000$ blocks sampled \emph{i.i.d.} and with equal probability from a selection of two blocks each ($D=2$). In subfigure $(\textsc{a})$, the blocks chosen are as in \cref{ex:standard_blocks}, while in subfigure $(\textsc{b})$ they are chosen to allow for the existence of edge modes. \enquote{Bandgap}, \enquote{shared pass band} and \enquote{hybridisation region} refer to the regions where both, neither, or one of the two block propagation matrices are hyperbolic, respectively (see \cite{ammari.barandun.ea2025Universal}).
    In both cases, the bulk of the spectrum is arranged in the regions where at least one of the block propagation matrices fails to be hyperbolic. Because in both subfigures the blocks are chosen such that the source-sink condition is always fulfilled, \cref{thm:saxonhutner} ensures that, in the infinite case, the spectral gap is exactly given by the bandgap. In line with \cref{thm:semiinfinite-johnson,thm:finite_spectra}, we can observe in the figure that this continues to hold in the finite case, apart from the possibility of edge modes.
    }
    \label{fig:saxonhutner_edgemode}
\end{figure}

\section{Setting and tools}\label{sec:setting}
We begin by introducing the notions of \emph{block disordered} systems of \emph{subwavelength resonators} and the main tools used to study them, including the capacitance matrix approximation and the propagation matrix formalism. 
To that end, this section will contain some results from \cite{disorder}, and we refer to this work for further details. For the remainder of this work, for any vector $\bm v\in \C^n$ or $\bm v \in \ell^2(\Z)$, $\bseq{v}{i}$ shall denote its $i$\textsuperscript{th} component and $\norm{\bm v} = \big(\sum_{i} \abs{\bseq{v}{i}}^2\big)^{1/2}$ the $\ell^2$-norm. 

\subsection{Subwavelength resonators}\label{sec:subwavelength setting}
The central systems of interest in this work are one-dimensional chains of subwavelength resonators (see \cite{ammari.davies.ea2024Functional,feppon.cheng.ea2023Subwavelength, cbms,disorder}). That is, we consider an array $\mc D=\bigsqcup_{i=1}^N (x_i^{\iL}, x_i^{\iR})$ consisting of $N$ resonators $D_i = (x_i^{\iL}, x_i^{\iR})$. We will denote $\ell_i = x_i^{\iR} - x_i^{\iL}$ for $1\leq i \leq N$ the \emph{length} of the $i$\textsuperscript{th} resonator and $s_i = x_{i+1}^{\iL}-x_{i}^{\iR}$ for $1\leq i \leq N-1$ the \emph{spacing} between the $i$\textsuperscript{th} and $(i+1)$\textsuperscript{th} resonators. 

The resonators are distinct from the background medium due to differing wave speeds and densities that are given by
\begin{equation}
    v(x) = \begin{cases}
        v_i &\text{if } x\in D_i,\\
        v &\text{if } x\in \mathbb{R}\setminus \mc D, 
    \end{cases} \quad
    \rho(x) = \begin{cases}
        \rho_b &\text{if } x\in D_i,\\
        \rho &\text{if } x\in \mathbb{R}\setminus \mc D.
    \end{cases}
\end{equation}
After further imposing a one-dimensional outgoing radiation condition, we obtain the following coupled system of Helmholtz equations for the resonant modes (see \cite[(1.6)]{feppon.cheng.ea2023Subwavelength}). The resonance problem is to find $\omega$ such that there is a nontrivial solution $u$ to

\begin{equation}
\label{eq:waveeq}
    \begin{dcases}
\ds \frac{\mathrm{d}^2}{\mathrm{d}x^2} u + \frac{\omega^2}{v^2} u = 0,  &\text{in }  
\mathbb{R}\setminus \mc D,\\
\ds \frac{\mathrm{d}^2}{\mathrm{d}x^2} u + \frac{\omega^2}{v_i^2} u = 0,  & \text{in }   D_i, i=1,\ldots, N,\\
 u\vert_{\iR}(x^{\iLR}_i) - u\vert_{\iL}(x^{\iLR}_i) = 0,                                                                 & \text{for } i=1, \ldots, N,\\
        \left.\frac{\dd u}{\dd x}\right\vert_{\iR}(x^{\iL}_i) - \frac{\rho_b}{\rho} \left.\frac{\dd u}{\dd x}\right\vert_{\iL}(x^{\iL}_i) = 0, &   \text{for } i=1, \ldots, N,          \\
        \left.\frac{\dd u}{\dd x}\right\vert_{\iL}(x^{\iR}_i) - \frac{\rho_b}{\rho} \left.\frac{\dd u}{\dd x}\right\vert_{\iR}(x^{\iR}_i) = 0, &   \text{for } i=1, \ldots, N,          \\
\bigg( \frac{\mathrm{d}}{\mathrm{d} |x|} - \i \frac{\omega}{v} \bigg) u(x) =0, & \text{for } |x| \text{ large enough,} 
\end{dcases}
\end{equation}
where for a function $w$ we denote by 
\begin{align*}
    w\vert_{\iL}(x) \coloneqq \lim_{s \downarrow 0} w(x-s) \quad \mbox{and} \quad  w\vert_{\iR}(x) \coloneqq \lim_{s \downarrow 0} w(x+s)
\end{align*}
if the limits exist. 

We are interested in the \emph{subwavelength high-contrast regime}. Namely, we denote by $\delta = \rho_b / \rho$ the \emph{contrast parameter} and consider the resonant frequencies  $\omega(\delta)$ with non-negative real part that satisfy
\begin{equation*}
    \omega(\delta) \to 0 \quad \text{ as } \quad \delta \to 0.
\end{equation*}

As shown in \cite{feppon.cheng.ea2023Subwavelength}, in this regime there exist exactly $N$ subwavelength resonant frequencies, characterised in leading order by a \emph{material matrix} $V$ that is diagonal and a tridiagonal \emph{capacitance matrix} $C$, respectively defined as
\begin{gather} 
    V = \diag \left(\frac{v_1^2}{\ell_1}, \dots, \frac{v_N^2}{\ell_N}\right) \in \mathbb{R}^{N\times N}, \label{eq:matmat}
 \end{gather}
 and    
    \begin{gather}
    C = \left(\begin{array}{cccccc}
         \frac{1}{s_1}& -\frac{1}{s_1} \\
         -\frac{1}{s_1}& \frac{1}{s_1}+\frac{1}{s_2}& -\frac{1}{s_2} \\
         & -\frac{1}{s_2} & \frac{1}{s_2}+\frac{1}{s_3}& -\frac{1}{s_3}\\
         &&\ddots&\ddots&\ddots \\
         &&&-\frac{1}{s_{N-2}}& \frac{1}{s_{N-2}}+\frac{1}{s_{N-1}}& -\frac{1}{s_{N-1}}\\
         &&&&-\frac{1}{s_{N-1}}&\frac{1}{s_{N-1}}
    \end{array}\right) \in \mathbb{R}^{N\times N}. \label{eq:capmat}
\end{gather}
The following results are from \cite{feppon.cheng.ea2023Subwavelength}. 
\begin{theorem}\label{thm:capapprox}
    Consider a system consisting of $N$ subwavelength resonators in $\mathbb{R}$. Then, there exist exactly $N$ subwavelength resonant frequencies $\omega(\delta)$ that satisfy $\omega(\delta)\to 0$ as $\delta\to 0$. Furthermore, the $N$ resonant frequencies are given by 
    \[
        \omega_i(\delta) = \sqrt{\delta\lambda_i} + \mc O (\delta),
    \]
    where $(\lambda_i)_{1\leq i\leq N}$ are the eigenvalues of the eigenvalue problem
    \[
        VC\bm u_i = \lambda_i \bm u_i.
    \]
    Furthermore, the corresponding resonant modes $u_i(x)$ are approximately given by 
    \[
        u_i(x) = \sum_{j=1}^N \bm u_i^{(j)}V_j(x) + \mc O (\delta),
    \]
    where $\bm u_i^{(j)}$ is the $j$\textsuperscript{th} entry of the vector $\bm u_i$ and $V_j(x)$, $j=1,\ldots, N$, are defined as the solution to
    \begin{align*}
        \begin{dcases}
          -\frac{\dd{^2}}{\dd x^2} V_j(x) =0, & x\in\mathbb{R}\setminus \mc D, \\
          V_j(x)=\delta_{ij},              & x\in D_i, \ i=1,\ldots, N,                         \\
          V_j(x) = \BO(1),                  & \mbox{as } \abs{x}\to\infty.
        \end{dcases}
        \label{eq: def V_i}
    \end{align*}
\end{theorem}

We can thus fully understand the subwavelength resonant modes by studying the eigenvalue problem for the \emph{generalised capacitance matrix} $\mc C \coloneqq VC$ and will often use $\lambda_i$ and $\omega_i$ interchangeably, where it is clear from context. Note that the eigenvalue problem for $\mc C$ plays an analogous role as the tight-binding model for Schr\"odinger operators. Therefore, we refer to it as the discrete model for the Helmholtz operator in \eqref{eq:waveeq}.

\subsection{Block disordered systems}
Since we are interested in the band structure of disordered systems, where local translation invariance is broken, we must develop ways to construct and describe such systems. As discussed in \cite{disorder}, one such way is to consider a finite number of distinct \enquote{building blocks} consisting of (possibly multiple) subwavelength resonators. Later, constructing disordered resonator arrays from simple building blocks will enable us to characterise the ``unusual'' spectral properties of the array by looking at the building blocks.

A subwavelength block disordered system is a chain of subwavelength resonators consisting of $M$ blocks of resonators $B_{\oo{j}}$ sampled accordingly to a sequence $$\chi \in \{1,\dots , D\}^M \eqqcolon \ldm$$ from a selection $B_1, \dots, B_D$ of $D$ distinct resonator blocks, arranged on a line. Each resonator block $B_j$ is characterised by a sequence of tuples $(v_1,\ell_1,s_1), \dots, (v_{\len(B_j)},\ell_{\len(B_j)},s_{\len(B_j)})$ that denote the wave speed, length, and spacing of each constituent resonator. Here, $\len(B_j)$ denotes the total number of resonators contained within the block $B_j$.

We will often abuse notation and write $\mc D = \bigcup_{j=1}^M B_{\oo{j}}$ to denote the resonator array constructed by sampling the blocks $B_1, \dots B_D$  according to $\chi\in \ldm$ and then arranging them in a line. Having thus constructed an array of subwavelength resonators, we can alternatively write $\mc D = \bigcup_{i=1}^N D_i$ in line with \cref{sec:subwavelength setting}.
Note that as $M$ denotes the total number of sampled blocks and $N$ the total number of resonators, we always have $M\leq N$.

\begin{example}\label{ex:standard_blocks}
    Simple but nontrivial disordered systems can be obtained by sampling from a set of just two blocks $B_1$ and $B_2$. A canonical example of this, which we will consider extensively in this work, is where $B_1=B_{single}$ denotes a single resonator block with $\len(B_{single}) = 1$ and $\ell_1(B_{single}) = s_1(B_{single})  = 2$ while $B_2=B_{dimer}$ is a dimer resonator block with $\len(B_{dimer}) = 2$ and $\ell_1(B_{dimer}) = \ell_2(B_{dimer})=1$ and $s_1(B_{dimer})=1, s_2(B_{dimer})=2$. We choose all wave speeds to be equal to $1$.
    An example of a single realisation corresponding to the sequence $\chi = (1,2,1)$ is illustrated in \cref{fig:disorderedsketch}.
\end{example}

Notably, when repeated periodically, both the single resonator and the dimer block share a lower band $[0,1]$. However, the dimer block possesses an additional upper band $[2,3]$, and the distinct properties of these building blocks affect the resonant frequencies of the resulting array. The spectrum of a randomly arranged array of these two blocks can be seen in \cref{fig:saxonhutner_edgemode}($\textsc{a}$). As can be observed, both the shared lower band $[0,1]$ and the dimer-only upper band $[2,3]$ are populated by eigenvalues, while there are no eigenvalues anywhere else. In the remainder of this paper, this observation will be made rigorous as a consequence of \cref{thm:saxonhutner,thm:semiinfinite-johnson,thm:finite_spectra}. 

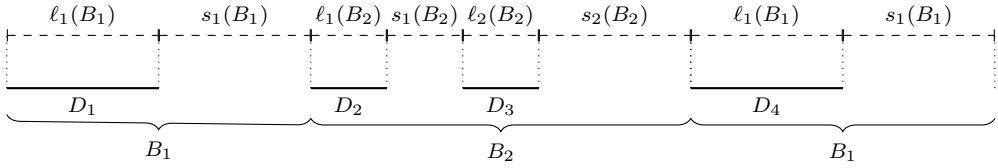
\begin{figure}
    \centering
    \begin{tikzpicture}[scale=1.0, every node/.style={font=\footnotesize}]

    \draw[thick] (0,0) -- (2,0);
    \node[below] at (1,0) {$D_1$};
    
    \draw[
    decorate,
    decoration={brace, mirror, amplitude=5pt}
    ] (0,-0.35) -- (4,-0.4)
    node[midway, yshift=-0.45cm] {$B_1$};
    
    \draw[-, dotted] (0,0) -- (0,0.7);
    \draw[-, dotted] (2,0) -- (2,0.7);
    \draw[|-|, dashed] (0,0.7) -- node[above]{$\ell_1(B_1)$} (2,0.7);
    
    \draw[-, dotted] (4,0) -- (4,0.7);
    \draw[|-|, dashed] (2,0.7) -- node[above]{$s_1(B_1)$} (4,0.7);
    
    \draw[thick] (4,0) -- (5,0);
    \node[below] at (4.5,0) {$D_2$};
    \draw[thick] (6,0) -- (7,0);
    \node[below] at (6.5,0) {$D_3$};
    
    \draw[
    decorate,
    decoration={brace, mirror, amplitude=5pt}
    ] (4,-0.4) -- (9,-0.4)
    node[midway, yshift=-0.45cm] {$B_2$};
    
    \draw[-, dotted] (5,0) -- (5,0.7);
    \draw[-, dotted] (6,0) -- (6,0.7);
    \draw[-, dotted] (7,0) -- (7,0.7);
    \draw[-, dotted] (9,0) -- (9,0.7);
    \draw[|-|, dashed] (4,0.7) -- node[above]{$\ell_1(B_2)$} (5,0.7);
    \draw[|-|, dashed] (5,0.7) -- node[above]{$s_1(B_2)$} (6,0.7);
    \draw[|-|, dashed] (6,0.7) -- node[above]{$\ell_2(B_2)$} (7,0.7);
    \draw[|-|, dashed] (7,0.7) -- node[above]{$s_2(B_2)$} (9,0.7);

    \begin{scope}[shift={(9,0)}]
    \draw[thick] (0,0) -- (2,0);
    \node[below] at (1,0) {$D_4$};
    \draw[
    decorate,
    decoration={brace, mirror, amplitude=5pt}
    ] (0,-0.4) -- (4,-0.4)
    node[midway, yshift=-0.45cm] {$B_1$};
    
    \draw[-, dotted] (0,0) -- (0,0.7);
    \draw[-, dotted] (2,0) -- (2,0.7);
    \draw[|-|, dashed] (0,0.7) -- node[above]{$\ell_1(B_1)$} (2,0.7);
    
    \draw[-, dotted] (4,0) -- (4,0.7);
    \draw[|-|, dashed] (2,0.7) -- node[above]{$s_1(B_1)$} (4,0.7);
    \end{scope}
\end{tikzpicture}
    \caption{A block disordered system consisting of two single resonator blocks $B_1$ and a dimer block $B_2$ arranged in a chain given by the sequence $\chi= (1,2,1)$. It thus consists of $M=3$ blocks and $N=4$ resonators $D_1,\dots ,D_4$ in total.}
    \label{fig:disorderedsketch}
\end{figure}
\subsection{Block sequences}\label{ssec:sampling}
Apart from the makeup of the individual blocks, the crucial factor determining the behaviour of block disordered systems is how these blocks are arranged. To enable later analysis, we will introduce the case of block sequences sampled independently and identically in detail. We consider a block distribution with sampling probabilities $p_1, \dots , p_D$, such that $\sum_{d=1}^D p_d=1$ and $p_d>0$ for all $d=1,\dots,D$.

We can make this sampling formal in a probability-theoretic sense.
\begin{definition}[Independently and identically sampled discrete sequences]
    We slightly overload notation to denote the discrete sequence space $D^{\Z} \coloneqq \{1,\dots, D\}^{\Z}$ and construct the probability space of \emph{independently and identically sampled discrete sequences} $(\ldz, \mc F, \prob)$ as follows. Let 
    \[
        C(i_1,\dots, i_n, X_1\times\dots \times X_n)\coloneqq \{\chi \in \ldz \mid (\oo{i_1},\dots, \oo{i_n}) \in X_1\times\dots \times X_n\}
    \] be the \emph{cylinder set} for $X_1,\dots, X_n \subset \{1,\dots, D\}$ and $i_1,\dots, i_n\in \Z$. On these cylinder sets, we define the probability measure
    \[
        \prob (C(i_1,\dots, i_n, X_1\times\dots \times X_n)) \coloneqq F(X_1)\cdot\dots\cdot F(X_n),
    \]
    where $F(X) = \sum_{j=1}^D p_j \mathbbm{1}_j(X)$ is the block probability distribution.
    Extending $\prob$ to the minimal $\sigma$-algebra containing all cylinder sets $\mc F$, we obtain the probability space $(\ldz, \mc F, \prob)$.

    Finally, we endow $\ldz = \{1,\dots D\}^{\Z}$ with the topology obtained by taking the countable product of the discrete topologies on $\{1,\dots D\}$. In this topology,  $\ldz$ is a compact metric space.
\end{definition}

We can extend this definition to finite sequences on $\ldm\coloneqq \{1,\dots, D\}^{M}$ by employing the \emph{truncation projection} $\mc P_M:\ldz \to \ldm$ to push forward the probability space $(\ldz, \mc F, \prob)$ to $(\ldm, \mc P_M^*\mc F, \mc P_M^*\prob)$. 

A key property of the family of \emph{i.i.d.} sequences is its metric transitivity under the group of shifts.
\begin{definition}[Metrically transitive group]
    Let $(\Omega, \mc F, \prob)$ be a probability space and $\mc S$ a topological group of automorphisms on $\Omega$. Then, $\mc S$ is called \emph{metrically transitive} if 
    \begin{enumerate}[label=(\roman*)]
        \item $\mc S$ is \emph{stochastically continuous}, \emph{i.e.}, $\lim_{S_1\to S} \prob(S_1X\triangle SX) = 0$ for any $S\in \mc S$ and $X\in \mc F$, where $\triangle$ denotes the symmetric difference;
        \item $\mc S$ is measure-preserving, \emph{i.e.}, $\prob(SX) = \prob(X)$ for any $S\in \mc S$ and $X\in \mc F$;
        \item any $\mc S$-invariant set $X\in \mc F$ has either probability $1$ or $0$.
    \end{enumerate}
    Here, $\mc S$-invariance of a set $X\in \mc F$ is understood to mean $SX = X$ for all $S\in \mc S$.
\end{definition}

The following follows immediately from the cylinder set construction.
\begin{lemma}
    The group of shifts $\mc S = \{S^j \mid j\in \Z\}$ where $(S^j\chi)(i) = \chi(i+j)$ is a metrically transitive group of automorphisms on $(\ldz, \mc F, \prob)$.
\end{lemma}

Another crucial notion is the idea of \emph{pseudo-ergodicity} of an element $\chi \in \ldz$.
\begin{definition}\label{def:pseudoergodic}
    Let $\mathbb{I} \in \{\mathbb{Z},\mathbb{N}\}$. We say that a sequence $\chi \in D^{\mathbb{I}}$ is \emph{pseudo-ergodic} if any finite sequence $\chi' \in D^M$, where $M \in \mathbb{N}$, is contained in $\chi$.
    
    More precisely, for any $\chi' \in D^M$, there must exist $j\in \mathbb{I}$ such that $\mc P_M S^j \chi = \chi'$. We will denote this by $\chi'\prec \chi$.
\end{definition}
Pseudo-ergodicity of $\chi$ is equivalent to the orbit $\operatorname{Orb}(\chi):=\{S^j\chi \mid j\in \Z\}$ being dense in $\ldz$, \emph{i.e.},
$$
\overline{\operatorname{Orb}(\chi)} = \ldz,
$$
and is a kind of genericity condition in the sense that it occurs with probability one for \emph{i.i.d.} sequences.

\begin{lemma}
    Any Bernoulli-distributed sequence $\chi\in \ldz$ is pseudo-ergodic with probability one.
\end{lemma}
\begin{proof}
    We define the sets
    \[
        \Psi_M \coloneqq\{\chi \in \ldz \mid \chi'\prec \chi \text{ for all }\chi'\in \ldm\},
    \] 
    which are $\mc S$-invariant and have $\prob(\Psi_M)>0$ for all $M\in \N$. The latter part follows after realising that finding all sequences $\chi'\in \ldm$ lined up sequentially starting at $0$ occurs with nonzero probability if $p_d>0$ for all $d=1,\dots, D$. By metric transitivity, we must thus have $\prob(\Psi_M)=1$ for all $M\in \N$. 

    On the other hand, we denote $\Psi := \{\chi\in \ldz \mid \chi \text{ pseudo-ergodic}\}$ and find 
    \[
        \Psi = \bigcap_{M=1}^\infty\Psi_M.
    \]
    But, due to the fact that $\prob(\Psi_M)=1$ for all $M$, we must also have $\prob(\Psi)=1$, as desired.
\end{proof}

\subsection{Resonator sequences}\label{ssec:resonator seq}
At this point, we would also like to connect the sampling of blocks to the corresponding sequence of resonators. Consider a block disordered system with blocks $B_1, \dots B_D$ sampled with probability $p_1, \dots, p_D$ yielding an \emph{i.i.d.} block sequence $\chi \in \ldz$. 
To encode the sequences of resonators, we use tuples from the set
\[
    R \coloneqq\left\{ (d,r) \mid d\in \{1,\dots D\}, r\in \{1,\dots, \len(B_d)\}\right\} \subset \N^2.
\]
The tuple $(d,r)$ denotes the $r$\textsuperscript{th} resonator of the $d$\textsuperscript{th} block. A block sequence $\chi\in\ldz$ determines a concatenated resonator sequence once a block boundary is chosen as the origin. We denote this boundary-pointed sequence by $\Phi_0(\chi)$, where the resonator at index zero is the first resonator of $B_{\chi(0)}$.

If the block lengths are not all equal, the image $\Phi_0(\ldz)$ is not invariant under the one-resonator shift. We therefore introduce the space
\[
 \widehat\Omega
 :=\{(\chi,r):\chi\in\ldz,\ 1\le r\le \len(B_{\chi(0)}) \}
\]
and define
\[
 \widehat S(\chi,r)
 :=
 \begin{cases}
  (\chi,r+1),&r<\len(B_{\chi(0)}),\\
  (S\chi,1),&r=\len(B_{\chi(0)}).
 \end{cases}
\]
Thus, $\widehat S$ advances the origin by one resonator. Let $\lrz \subset R^{\Z}$ denote the space of all bi-infinite labeled resonator sequences obeying the block-transition rules. There is a natural bijection
\[
 \Phi:\widehat\Omega\longrightarrow \lrz
\]
such that the resonator at the origin of $\Phi(\chi,r)$ is the $r$\textsuperscript{th} resonator of $B_{\chi(0)}$. 
If $S_R$ denotes the one-resonator shift, then
\[
 \Phi \widehat S=S_R \Phi.
\]
The original boundary-pointed map is $\Phi_0(\chi)=\Phi(\chi,1)$.
Let $\mu$ be the Bernoulli block measure and define a probability measure $\widehat\mu$ on $\widehat\Omega$ by
\[
 \int_{\widehat\Omega} f\,d\widehat\mu
 :=\frac1{\sum_{d=1}^D p_d \len(B_d)}
 \int_{\ldz}\sum_{r=1}^{\len(B_{\chi(0)})} f(\chi,r)\,d\mu(\chi).
\]
A direct calculation using the shift invariance of $\mu$ shows that $\widehat\mu$ is invariant under $\widehat S$. The stationary resonator measure is
\[
 \mu_{R}:= \Phi^*\widehat\mu.
\]

\begin{example}
    Consider a block disordered system with blocks that are either a single resonator $B_1 = B_{single}$ or a dimer block $B_2 = B_{dimer}$, as in \cref{ex:standard_blocks}. We sample the blocks \emph{i.i.d.} with $B_1$ and $B_2$ occurring with probability $p_1$ and $p_2= 1-p_1$, respectively, yielding a sequence $\chi \in 2^{\Z}$. 
    The corresponding resonator sequence $\alpha = \Phi_0(\chi) =\Phi(\chi,1) \in \lrz$ consists of the tuples $R = \{(1,1), (2,1), (2,2)\}$ encoding the single resonator and the two dimer resonators, respectively. The finite state space $R$ then allows us to collect the transition probabilities in the \emph{transition matrix}
    \[
    P =\begin{pmatrix}
        p_1 & p_2 & 0\\
        0 & 0 & 1\\
        p_1 & p_2 & 0
    \end{pmatrix},
    \]
    where $P_{ij}$ is the transition probability from the $i$\textsuperscript{th} state to the $j$\textsuperscript{th} state and we encode the block tuples $(1,1), (2,1), (2,2)$ as states $1,2,3$, respectively. Note that the transition matrix $P$ has stationary distribution proportional to $(p_1,p_2,p_2)$, namely, $(p_1,p_2,p_2)/(p_1 +2p_2)$, not the distribution at the boundary-pointed origin obtained directly from $\chi(0)$.
\end{example}

\begin{lemma}
 The one resonator shift $S_R$ is measure preserving and ergodic on 
 $(\lrz, \mu_R)$. Consequently, the group of shifts $\mc S_R = \{S_R^j \mid j\in \Z\}$ is a metrically transitive group of automorphisms.
    
\end{lemma}

\begin{proof}
This result follows from \cite[Example 1.15b)]{pastur1992Spectra} where the essential ingredients are the fact that we consider homogeneous Markov chains and that the finite state space together with irreducibility ensures the uniqueness of the stationary state.
\end{proof}

We equip \(\lrz \) with the subspace topology
inherited from the product topology on \(R^{\mathbb Z}\).
It is a compact metric space, and
\(\Phi:\widehat\Omega\to \lrz \)
is a homeomorphism.
We say that
\(\alpha\in \lrz \) is pseudo-ergodic if every
finite admissible resonator word occurs in \(\alpha\). More precisely,
if \(\mathbb{L}_{R}(M)\) denotes the set of admissible resonator sequences
of length \(M\), then
\[
\alpha\text{ is pseudo-ergodic}
\quad\Longleftrightarrow\quad
\alpha'\prec\alpha
\quad\text{for every }M\in\mathbb N
\text{ and }\alpha'\in \mathbb{L}_{R}(M).
\]
Equivalently,
\[
\alpha\text{ is pseudo-ergodic}
\quad\Longleftrightarrow\quad
\overline{\operatorname{Orb}_{R}(\alpha)}
= \lrz,
\]
where
\[
\operatorname{Orb}_{R}(\alpha)
:=
\{S_{R}^j\alpha  \mid j\in\mathbb Z\}.
\]
Furthermore, for every
\((\chi,r)\in\widehat\Omega\),
\[
\chi\text{ is pseudo-ergodic}
\quad\Longleftrightarrow\quad
\Phi(\chi,r)\text{ is pseudo-ergodic}.
\]
Notice that the finite sequences appearing here are only the admissible
sequences in \(\mathbb{L}_{R}(M)\), not all sequences in
\(R^M\).

Finally, because the stationary Markov shift is ergodic and every
nonempty admissible cylinder has a positive measure,
\(\mu_{R}\)-almost every resonator sequence is
pseudo-ergodic.


\section{Jacobi operators and cocycles}\label{sec:jacobi}
In the infinite block sequence limit, the subwavelength resonance problem is no longer described by a capacitance matrix $\mc C = VC$ but a real Jacobi operator on $\ell^2(\Z)$.
\begin{definition}[Jacobi Operator]
    A (real) \emph{Jacobi operator} is a symmetric tridiagonal operator acting on $\lz$. It can be described by two real sequences $\bseq{a}, \bseq{b} \in \ell^\infty(\Z)$ and is given by
    \begin{align*}
       \mathcal{J}:\lz &\to \lz\\
        \bseq{v}{i} &\mapsto (J\bm v)^{(i)} = \bseq{a}{i-1}\bseq{v}{i-1} + \bseq{a}{i}\bseq{v}{i+1} + \bseq{b}{i}\bseq{v}{i} .
    \end{align*}
\end{definition}
Because $\bseq{a}{i}$ and $\bseq{b}{i}$ are bounded and $\mathcal{J}$ is tridiagonal and symmetric, $J:\lz\to \lz$ is a well-defined, linear, bounded, and self-adjoint operator.

We note that the symmetry requirement on $\mathcal{J}$ does not pose any problems, since the capacitance matrix eigenvalue problem for $VC$ demonstrated in \cref{thm:capapprox} is equivalent to the symmetric eigenvalue problem for $V^\frac{1}{2}CV^\frac{1}{2}\sim VC$:
$$
VC =V^\frac{1}{2} \bigg( V^\frac{1}{2} C V^\frac{1}{2} \bigg) V^{-\frac{1}{2}}.
$$

\begin{remark}
    As we aim to understand the spectrum of large block disordered systems, we investigate the symmetrised capacitance matrix $V^\frac{1}{2}CV^\frac{1}{2}$ in the large resonator array asymptotic $N\to \infty$. Here, $N\to \infty$ is understood as starting with an infinite array of subwavelength resonators (not necessarily arranged in blocks) described by sequences of wave speeds $(v_i)_{i=-\infty}^\infty$, resonator lengths $(\ell_i)_{i=-\infty}^\infty$ and spacings $(s_i)_{i=-\infty}^\infty$ and constructing a sequence of increasingly large capacitance matrices $\mc J_N \in \mathbb{R}^{(2N+1)\times (2N+1)}$ from the truncated sequences $(v_i)_{i=-N}^N, (\ell_i)_{i=-N}^N, (s_i)_{i=-N}^N$. 

    If we consider $(\mc J_N)_{N=1}^\infty$ with $\mc J_N := V^\frac{1}{2} C V^{\frac{1}{2}}$ as a sequence of operators on the sequence space $\ell^2(\Z)$, we find that this sequence converges strongly to a Jacobi operator $\mc J:\ell^2(\Z)\to \ell^2(\Z)$ with the following sequences:
    \begin{equation}\label{eq:jacobibands}
        \bseq{a}{i} = -\frac{v_{i}v_{i+1}}{s_{i}\sqrt{\ell_{i}\ell_{i+1}}} \quad \text{and} \quad 
        \bseq{b}{i} = \frac{v_i^2}{\ell_i}\left(\frac{1}{s_{i-1}}+\frac{1}{s_{i}}\right).
    \end{equation}
\end{remark}

Given the blocks $B_1, \dots, B_D,$ any sequence $\chi\in \ldz$ uniquely determines a sequence of resonators $\alpha \in \lrz$ and thus also unique sequences of wave speeds $(v_i)_{i=-\infty}^\infty$, resonator lengths $(\ell_i)_{i=-\infty}^\infty$ and spacings $(s_i)_{i=-\infty}^\infty$. Condition \cref{eq:jacobibands} thus allows us to define a Jacobi operator $\mc{J}(\alpha)$ for any resonator sequence $\alpha$, making $\mc{J}:\lrz\to \mc L(\ell^2(\Z))$ a \emph{random operator} on the probability space of the resonator sequences $(\lrz, \mc B(\lrz), \mu_R)$. We will often overload the notation to mean $\mc J(\chi)\coloneqq \mc J(\Phi_0(\chi))$. Note that in the particular case where the sequences $(\bseq{a}{i})$ and $(\bseq{b}{i})$ are given by \eqref{eq:jacobibands}, it is easy to see that $\mc{J} \geq 0$ and therefore, $\sigma(\mc J) \subset [0, +\infty)$.

We have the following result from \cite[Proposition 2]{alkorn.zhang2022correspondence} ensuring the almost sure invariance of the spectrum
$\sigma(\mc J(\alpha))$.

\begin{proposition}\label{prop:spec-invar-as}
    If $\alpha\in \lrz$ is pseudo-ergodic, then we have
    \[
        \sigma(\mc J(\alpha'))\subset \sigma(\mc J(\alpha))
    \]
    for \emph{any} $\alpha'\in \lrz$. 
\end{proposition}

\subsection{Transfer matrix cocycles}
To study the spectral properties of such tridiagonal Jacobi operators $\mc J$, it will be very fruitful to investigate the solutions of the \emph{spectral equation}
\begin{equation}\label{eq:spectraleq}
    (\mc J-\lambda)\bm v = 0,
\end{equation}
for $\lambda \geq 0$ and some formal solution $\bm v\in \C^{\Z}$, not necessarily $\ell^2$-summable or even bounded. 
Note that negative $\lambda$ are not relevant here because $\sigma(\mc J) \subset [0, +\infty)$.

We can exploit the tridiagonal structure to obtain the following iterative characterisation of $\bm v$.
\begin{definition}
    For any resonator sequence $\alpha\in \lrz$ and any parameter $\lambda\geq 0$, we define the \emph{transfer matrices} (also known as the \emph{Jacobi cocycle map}) as
    $T^\lambda:\Z \to \on{GL}(2,\R)$ given by 
    \begin{equation}\label{eq:transfermateq}
        T^\lambda(i) \coloneqq \frac{1}{\bseq{a}{i}}\begin{pmatrix}
            \lambda - \bseq{b}{i} & -\bseq{a}{i-1}\\
            \bseq{a}{i} & 0
        \end{pmatrix}
    \end{equation}
    for any $i\in \Z$, where $\bseq{a}{i}$ and $\bseq{b}{i}$ are defined by \eqref{eq:jacobibands}.

    Notably, for any $\bm v\in \C^{\Z}$ solving \cref{eq:spectraleq}, we have 
    \begin{equation}\label{eq:transfermatprop}
        \begin{pmatrix}
            \bseq{v}{i+1}\\
            \bseq{v}{i}
        \end{pmatrix} = T^\lambda(i) \begin{pmatrix}
            \bseq{v}{i}\\
            \bseq{v}{i-1}
        \end{pmatrix}
    \end{equation}
    for all $i\in \Z$.

    We further define the \emph{cocycle iteration} as 
    \begin{equation}\label{eq:transfermat_cocycle_iteration}
        T^\lambda_n(i) \coloneqq \begin{cases}
            T^\lambda(i+n-1)\cdot \cdot \cdot T^\lambda(i), & n\geq 1,\\
            I_2, & n=0,\\
            (T^\lambda(i+n))^{-1}\cdot \cdot \cdot (T^\lambda(i-1))^{-1}, & n\leq-1,
        \end{cases}
    \end{equation}
    and find 
    \begin{equation}
        \begin{pmatrix}
            \bseq{v}{i+n}\\
            \bseq{v}{i+n-1}
        \end{pmatrix} = T_n^\lambda(i) \begin{pmatrix}
            \bseq{v}{i}\\
            \bseq{v}{i-1}
        \end{pmatrix}
    \end{equation}
    for all $n\in \Z$.
\end{definition}
We note that, by \cref{eq:jacobibands}, we have $\bseq{a}{i}\neq 0$ for all $i\in \Z$ and therefore, \cref{eq:transfermateq} is well defined.

It will often prove useful to think of the transfer matrices acting on the \emph{real projective space}. 
\begin{definition}[Real projective space]
Consider the equivalence relation $\sim$ on $\R^2$ defined by $\bm u \sim \bm v$ if and only if $\bm u = c \bm v$ for $\bm u, \bm v \in \R^2$ and $c \in \R \setminus \{0\}$. The \emph{real projective space} (of dimension one) $\rp$ is then defined as the quotient $\R^2 / \sim$. For any $\bm u\in \R^2$, we will denote by $\overline{\bm u}$ the equivalence class of $\bm u$ under $\sim$.

It is clear that there exists a one-to-one correspondence between the one-dimensional subspaces of $\R^2$ and the elements of $\rp$. 
Furthermore, the group $\operatorname{GL}(2,\R)$ acts on $\rp$ via the matrix-vector product of representatives, \emph{i.e.},  $A\overline{u} \coloneqq \overline{Au}$.

We can define a metric on $\rp$ by 
\begin{equation}\label{eq:prmetric}
\begin{aligned}
    d:\rp\times \rp &\to \R \\
    (\overline{\bm u}, \overline{\bm v})&\mapsto \sqrt{1-{\left(\frac{\ip{\bm u}{\bm v}}{\norm{\bm u}\norm{\bm v}}\right)}^2} = \sin \measuredangle(\overline{\bm u}, \overline{\bm v}),
\end{aligned}
\end{equation}
where $\measuredangle(\overline{\bm u}, \overline{\bm v})$ denotes the angle between $\overline{\bm u}$ and $\overline{\bm v}$.
\end{definition}

A crucial concept relating the properties of the transfer matrix cocycle to the spectrum of $\mc J$ is the notion of \emph{dominated splitting}, adapted from \cite{alkorn.zhang2022correspondence}. 

\begin{definition}[Dominated splitting]\label{def:DS}
We say that
$T\in \ell^\infty(\Z,\operatorname{GL}(2,\R))$
admits a \emph{dominated splitting} (DS) if there exist two maps $s,u:\Z\to \rp$ such that
\begin{enumerate}[label=(DS\arabic*)]
  \item \textbf{$T$-invariance:}  
        For every $i\in\mathbb{Z}$,
        \[
          T(i)u(i)=u(i+1)
          \quad\text{and}\quad
           T(i)s(i)=s(i+1);
        \]

  \item \textbf{Domination:}  
        There exist $N\in\mathbb{N}$ and $\eta>1$ such that
        \[
          \bigl\lVert T_{N}(i)\,\vec{u}(i)\bigr\rVert
          \;>\;
          \eta\,
          \bigl\lVert T_{N}(i)\,\vec{s}(i)\bigr\rVert
        \]
        for all \(i\in\mathbb{Z}\) and all unit vectors
        \(\vec{u}(i)\in u(i)\) and
        \(\vec{s}(i)\in s(i)\), where $$T_N(i) = T(i+N-1)\cdot \cdot \cdot T(i);$$

  \item \textbf{Angle separation:}  
        There exists $\delta>0$ such that
        \[
          d\bigl(u(i),s(i)\bigr)>\delta
          \quad\text{for all } i\in\mathbb{Z},
        \]
where $d$ is defined by \eqref{eq:prmetric};
  \item \textbf{Non-degeneracy:}  
        For the $N\in\mathbb{N}$ chosen in \textup{(DS2)}, we have
        \[
          \inf_{i\in\mathbb{Z}}\,
          \norm{T_{N}(i)}>0.
        \]
\end{enumerate}

In this case, we write $T\in \mc{DS}$.
\end{definition}

The reason why the existence of a dominated splitting is so crucial is that it completely determines the spectral gaps of $\mc J$; see \cite[Theorem 4]{alkorn.zhang2022correspondence}.
\begin{theorem}[Dominated splitting version of Johnson's theorem]\label{thm:johnson}
    \[
        \lambda\notin\sigma(\mc J) \iff T^\lambda\in \mc{DS}.
    \]
\end{theorem}
This, together with \cref{prop:spec-invar-as}, gives an almost-sure characterisation of the spectrum for any resonator sequence. 

Intuitively, this theorem holds because if $T^\lambda\in \mc{DS}$, the stable and unstable directions allow the construction of a Green function, ensuring that $\mc J - \lambda$ is invertible. Conversely, if $\lambda$ lies in the resolvent of $\mc J$, one may use the fact that, by a Combes-Thomas-type estimate, the Green function decays exponentially to construct a dominated splitting.

The remainder of this section will consist of a series of reductions enabling us to find simple necessary and sufficient conditions for the existence of a dominated splitting in terms of block properties. An overview of these reductions is given in \cref{eq:saxon-hutner_schematic}.

\subsection{Propagation matrix cocycle}
The transfer matrix cocycle has two key weaknesses. First, we have in general $$\det T^\lambda(i) = \bseq{a}{i-1} / \bseq{a}{i} \neq 1,$$ which prevents the use of the theory of $\on{SL}(2,\R)$-cocycles from being used. And secondly, the $i$\textsuperscript{th} transfer matrix $T^\lambda(i)$ depends not only on the properties of the $i$\textsuperscript{th} resonator but also on the surrounding resonators. This complicates the association of the resonator properties with the transfer matrix properties.

To remedy this, we perform a change of basis at the transfer matrix level and switch to propagation matrices.
\begin{definition}[Propagation matrix]\label{def:propmat}
    We define the \emph{conjugacy sequence} $Q:\Z\to \on{GL}(2,\R)$ as 
    \begin{equation}\label{eq:proptransferhomology}
        Q(i) \coloneqq \begin{pmatrix}
          \ds  \frac{v_{i}}{\sqrt{\ell_{i}}} & 0 \\
           \ds \frac{v_{i}}{\ell_{i}s_{i-1}} & \ds -\frac{v_{i-1}}{\ell_{i-1}s_{i-1}}
        \end{pmatrix}.
    \end{equation}
    The \emph{propagation matrix cocycle} is then defined as $P^\lambda:\Z\to \on{SL}(2,\R)$
    \begin{equation}\label{eq:propmat}
        P^\lambda(i) \coloneqq Q(i+1)T^\lambda(i)(Q(i))^{-1} = \begin{pmatrix}
          \ds   1 - s_i\frac{\ell_i}{v_i^2}\lambda & s_i\\
           \ds  - \frac{\ell_i}{v_i^2}\lambda & 1
        \end{pmatrix},
    \end{equation}
    and the cocycle iteration $P^\lambda_n(i)$ is defined analogously to $T^\lambda_n(i)$.
\end{definition}

This relation between $T^\lambda$ and $P^\lambda$ is a cocycle cohomology.
\begin{definition}[Cocycle cohomology]
    We say that two cocycles $A,B:\Z\to \on{GL}(2,\R)$ are \emph{cohomologous} if there exists a conjugacy sequence $Q:\Z\to \on{GL}(2,\R)$ such that
    \begin{equation}
        A(i) = Q(i+1)B(i)(Q(i))^{-1}
    \end{equation}
    for all $i\in \Z$.
\end{definition}

We note that cocycle cohomology interacts nicely with the cocycle iteration, and we have 
\begin{equation}
    P^\lambda_n(i) = Q(i+n)T_n^\lambda(i)(Q(i))^{-1}
\end{equation}
for all $n\in \Z$. From \cref{eq:propmat}, it is immediately clear that the propagation matrix remedies the two issues of the transfer matrix outlined above. That is, it always has $\det P^\lambda(i)=1$ and depends only on the properties $\ell_i,s_i,v_i$ of the $i$\textsuperscript{th} resonator.

\begin{remark}\label{rmk:propmat_intuition}
    In the literature (see, for instance, \cite[Chapter 11]{book-lukic}), $$\det T^\lambda(i) = \bseq{a}{i-1} / \bseq{a}{i} \neq 1$$ is often overcome by choosing the conjugacy sequence 
    \[
        B(i) = \begin{pmatrix}
            1 & 0\\
            0 & \bseq{a}{i}
        \end{pmatrix}
    \]
    corresponding to a change of basis $(\bseq{v}{i+1},\bseq{v}{i})^\top \to (\bseq{v}{i+1},\bseq{a}{i}\bseq{v}{i})^\top$.
    
    The conjugacy sequence $Q$ is based on a similar idea with clear physical intuition:
    We can split $Q(i) = D(i)V(i)$ with 
    \[
        D(i)= \begin{pmatrix}
            1 & 0\\
         \ds   \frac{1}{s_{i-1}} & \ds -\frac{1}{s_{i-1}}
        \end{pmatrix}, \quad V(i)= \begin{pmatrix}
          \ds   \frac{v_{i}}{\sqrt{\ell_{i}}} & 0\\
            0& \ds \frac{v_{i-1}}{\sqrt{\ell_{i-1}}}
        \end{pmatrix}.
    \]
  The matrix $V(i)$ then corresponds to the change of basis $(\bseq{v}{i},\bseq{v}{i-1})^\top \to (\bseq{u}{i},\bseq{u}{i-1})^\top$ undoing the similarity transform $V^\frac{1}{2}CV^\frac{1}{2}\sim VC$ that was performed to obtain the Jacobi operator $\mc J$ from the capacitance problem $VC$. Following \cref{thm:capapprox}, $\bseq{u}{i}$ can thus physically be understood to be the first-order approximant of the solution $u(x_i^{\iL})$ of \cref{eq:waveeq} on the resonator $D_i$. 
    
    Now, to remedy the non-locality of $T^\lambda$, $D(i)$ performs a change of basis from $$(\bseq{u}{i},\bseq{u}{i-1})^\top = (u(x_{i}^{\iL}),u(x_{i-1}^{\iL}))^\top \to (u(x_{i}^{\iL}),u'(x_{i}^{\iL}))^\top,$$ \emph{i.e.}, to the tuple $(u(x_{i}^{\iL}),u'(x_{i}^{\iL}))^\top$ encoding $u(x_{i}^{\iL})$ together with the exterior derivative $u'(x_{i}^{\iL}) = \lim_{x\uparrow x_{i}^{\iL}} \frac{d}{dx}u(x)$. Here, the form of $D(i)$ again follows from the first-order approximation of $u(x)$ obtained in \cref{thm:capapprox}.
\end{remark}

For a propagation matrix cocycle in $\on{SL}(2,\R)$, we can now define the concept of uniform hyperbolicity \cite{zhang2020Uniform}, which is a strengthened version of the dominated splitting for $\on{SL}(2,\R)$-matrices. 

\begin{definition}[Uniform hyperbolicity]\label{def:UH}
We say that a cocycle
$P\in \ell^\infty(\Z,\operatorname{SL}(2,\R))$
is \emph{uniformly hyperbolic} (UH) if
there exist two maps $s,u:\Z\to \rp$ such that
\begin{enumerate}[label=(UH\arabic*)]
  \item \textbf{$P$-invariance:}  
        For every $i\in\mathbb{Z}$,
        \[
          P(i)u(i)=u(i+1)
          \quad\text{and}\quad
           P(i)s(i)=s(i+1);
        \]
  \item \textbf{Uniform growth:}  
        There exist $C>0$ and $\eta>1$ such that
        \[
          \norm{P_{-n}(i)\vec{u}(i)}, \norm{P_n(i)\vec{s}(i)} \leq C\eta^{-n}
        \]
        for all $i\in\mathbb{Z}$, $n\in \N$ and all unit vectors
        $\vec{u}(i)\in u(i)$ and
        $\vec{s}(i)\in s(i)$.
\end{enumerate}

In this case, we write $P\in \mc{UH}$.
\end{definition}

Indeed, for $\on{SL}(2,\R)$-matrices, uniform hyperbolicity and dominated splitting are equivalent.
\begin{lemma}\label{lem:DSiffUH}
    Consider a cocycle $P:\Z\to \on{SL}(2,\R)$. Then
    \[
        P\in \mc{UH} \iff P\in \mc{DS}.
    \]
\end{lemma}
\begin{proof}
    \enquote{$\implies$}: Let $P\in \ell^\infty(\Z,\operatorname{SL}(2,\R))$ uniformly hyperbolic. Then conditions (DS1) and (DS2) follow immediately, (DS4) follows from the fact that $P(i)\in \on{SL}(2,\R)$, and (DS3) is \cite[Lemma 2]{zhang2020Uniform}.
    
    \enquote{$\impliedby$}: We now consider  $P\in \ell^\infty(\Z,\operatorname{SL}(2,\R))$ admitting a dominated splitting $u, s:\Z \to \rp$. We want to prove that $P$ is also uniformly hyperbolic (not necessarily with the same $u,s$) and use \cite[Theorem 1]{zhang2020Uniform} to equivalently prove that there exists $c>0, \xi>1$ such that 
    \[
        \norm{P_n(i)}\geq c\xi^n
    \] for all $n\in \N$ and $i\in \Z$.

    We can repeatedly apply (DS2) to find\\
   
    \textbf{Claim:}
    \[
        \norm{P_{KN}(i)u(i)} > \eta^K\norm{P_{KN}(i)s(i)} \quad \text{for all }i \in \Z, K\in \N .
    \]
    \begin{proof}
        We proceed by induction and let $i\in \Z$. The case $K=1$ follows by (DS2).
        For the induction step, we again use (DS2) with $i' = i + (K-1)N$ to have 
        \[
            \norm{P_N(i+(K-1)N)\vec{u}(i+(K-1)N)} > \eta \norm{P_N(i+(K-1)N)\vec{s}(i+(K-1)N)}.
        \]
        Using (DS1) repeatedly, we find
        \[
            \vec{u}(i+(K-1)N) = \frac{P_{(K-1)N}(i)\vec{u}(i)}{\norm{P_{(K-1)N}(i)\vec{u}(i)}}, \quad \vec{s}(i+(K-1)N) = \frac{P_{(K-1)N}(i)\vec{s}(i)}{\norm{P_{(K-1)N}(i)\vec{s}(i)}},
        \] and plugging this into the above inequality yields
        \[
            \frac{\norm{P_{KN}(i)\vec{u}(i)}}{\norm{P_{(K-1)N}(i)\vec{u}(i)}} > \eta \frac{\norm{P_{KN}(i)\vec{s}(i)}}{\norm{P_{(K-1)N}(i)\vec{s}(i)}}.
        \] 
        Now, by the induction hypothesis, we have 
        \[
            \frac{\norm{P_{KN}(i)\vec{u}(i)}}{\eta^{K-1}\norm{P_{(K-1)N}(i)\vec{s}(i)}} > \frac{\norm{P_{KN}(i)\vec{u}(i)}}{\norm{P_{(K-1)N}(i)\vec{u}(i)}} > \eta \frac{\norm{P_{KN}(i)\vec{s}(i)}}{\norm{P_{(K-1)N}(i)\vec{s}(i)}},
        \] and the claim follows after multiplying both sides of the inequality by $\eta^{K-1}\norm{P_{(K-1)N}(i)\vec{s}(i)}$.
    \end{proof}

    For any $K\in \N$ and $i\in \Z$, we denote $\sigma_1\geq \sigma_2 \geq 0$ the singular values of $P_{KN}(i)$ and have 
    \[
        \sigma_1 \geq \norm{P_{KN}(i)u(i)} > \eta^K\norm{P_{KN}(i)s(i)} \geq \eta^K\sigma_2.
    \]
    Crucially, because $\det P(i) = 1$ for all $i$, we also have $\det P_{KN}(i) = 1$ and thus $$\sigma_1\sigma_2 = \abs{\det P_{KN}(i)} = 1.$$ Plugging this into the inequality above yields 
    \[
        \sigma_1 = \norm{P_{KN}(i)} > \sqrt{\eta}^K,
    \]
    ensuring that the norm of $P_{KN}(i)$ grows exponentially in $K$.
    
    Finally, because $P\in \ell^\infty(\Z,\operatorname{SL}(2,\R))$, we know that there must exist some bound $C$ such that for all $i\in \Z$ $\sigma_1 < C$ and thus also $\sigma_2>1/C$ where $\sigma_1\geq \sigma_2$ again denote the singular values of $P(i)$. Because $N$ is independent of $i$, this ensures that $\norm{P_n(i)}$ cannot arbitrarily decay between the exponential growth points $P_{KN}(i)$, allowing us to find $c>0$ and $\xi>1$ such that $\norm{P_n(i)}\geq c\xi^n$.
\end{proof}

Furthermore, the existence of a dominated splitting is invariant under a sufficiently regular cohomology.
\begin{lemma}\label{lem:DSiffDS_cohomology}
    Consider two cocycles $P,T:\Z \to \on{GL}(2,\R)$ cohomologous by $Q:\Z \to \on{GL}(2,\R)$, that is,
\begin{equation} \label{eq:conjugacy}
P(i)= Q(i+1) T(i) Q(i)^{-1}, \quad \forall\; i \in \mathbb Z. 
\end{equation}

        Assume further that the conjugacy sequence $Q(i)$ is uniformly bounded from above and below, \emph{i.e.}, there exist $C_1, C_2\in \R$ such that $$C_1> \sigma_1(Q(i))\geq \sigma_2(Q(i))>C_2>0$$ for all $i\in \Z$, where $\sigma_{1,2}(Q(i))$ denote the upper and lower singular values of $Q(i)$, respectively.
    Then
    \[
        P\in \mc{DS} \iff T\in \mc{DS}.
    \] 
\end{lemma}
\begin{proof}
It suffices to prove $T\in\mc{DS} \Rightarrow P\in\mc{DS}$, since the converse follows from the same argument applied to
\[
 T(i)=Q(i+1)^{-1}P(i)Q(i)
\]
and the uniformly conditioned sequence $Q(i)^{-1}$.

Suppose that $T$ has invariant unstable and stable directions $u(i),s(i)\in\rp$.  Define
\[
 q(i):=Q(i)u(i),\qquad p(i):=Q(i)s(i).
\]
Equation~\eqref{eq:conjugacy} immediately gives $P(i)q(i)=q(i+1)$ and $P(i)p(i)=p(i+1)$.

Let $\vec u(i)$ and $\vec s(i)$ be unit representatives, and set
\[
 \alpha_i:=\|Q(i)\vec u(i)\|,\qquad
 \beta_i:=\|Q(i)\vec s(i)\|.
\]
Then $C_2\le\alpha_i,\beta_i\le C_1$, and unit representatives of the transformed directions are
\begin{equation}
 \vec q(i)=\frac{Q(i)\vec u(i)}{\alpha_i},\qquad
 \vec p(i)=\frac{Q(i)\vec s(i)}{\beta_i}.
\label{eq:normalized-directions}
\end{equation}
If $N$ and $\eta>1$ are the domination constants for $T$, iteration gives
\[
 \|T_{KN}(i)\vec u(i)\|>\eta^K\|T_{KN}(i)\vec s(i)\|.
\]
Since
\[
 P_{KN}(i)=Q(i+KN)T_{KN}(i)Q(i)^{-1},
\]
we find that
\[
 \|P_{KN}(i)\vec q(i)\|
 \ge\frac{C_2}{\alpha_i}\|T_{KN}(i)\vec u(i)\|
\]
and
\[
 \|P_{KN}(i)\vec p(i)\|
 \le\frac{C_1}{\beta_i}\|T_{KN}(i)\vec s(i)\|.
\]
Consequently, we obtain
\[
 \frac{\|P_{KN}(i)\vec q(i)\|}{\|P_{KN}(i)\vec p(i)\|}
 >\frac{C_2\beta_i}{C_1\alpha_i}\eta^K
 \ge\left(\frac{C_2}{C_1}\right)^2\eta^K.
\]
Choose $K$ such that $\eta':=(C_2/C_1)^2\eta^K>1$. This proves domination for $P$.

For angle separation, use
\[
 d([v],[w])=\frac{|\det(v,w)|}{\|v\|\|w\|}.
\]
Then
\[
 d(q(i),p(i))
 =\frac{|\det Q(i)|\,|\det(\vec u(i),\vec s(i))|}
 {\|Q(i)\vec u(i)\|\,\|Q(i)\vec s(i)\|}
 \ge\left(\frac{C_2}{C_1}\right)^2d(u(i),s(i)).
\]
Thus, uniform transversality is preserved.

Finally, using the equivalent condition
\[
 \inf_{i\in\Z}\|T_n(i)\|>0
\quad\text{for every fixed }n,
\]
we obtain
\[
 \|P_n(i)\|=\|Q(i+n)T_n(i)Q(i)^{-1}\|
 \ge\frac{C_2}{C_1}\|T_n(i)\|.
\]
This proves non-degeneracy and completes the implication.  Applying the argument with $Q^{-1}$ proves the converse.

\end{proof}

We have thus successfully reduced the question of whether $\lambda$ lies in the spectrum of $\mc J(\alpha)$ to the question of whether the corresponding propagation matrix cocycle is uniformly hyperbolic. 

\subsection{Block propagation matrices and the invariant cone criterion}
We now aim to develop simple and concrete criteria for the uniform hyperbolicity of the propagation matrix cocycle. Unfortunately, this is complicated by the fact that the transitions in the resonator sequence $\alpha\in \lrz$ are restricted by block construction as described in \cref{ssec:resonator seq}. 
Thus, it would be much more convenient to consider the \emph{block sequences} $\chi \in \ldz$ where, due to the \emph{i.i.d.} sampling, all transitions are allowed.

In fact, this is possible by grouping the propagation matrices by blocks to construct a new cocycle.

\begin{definition}\label{def:blockprop}
    For every block $B_d$, $d=1,\dots, D,$ we define the block propagation matrix
    \begin{equation}
        \mc P_d^\lambda \coloneqq \prod_{k=1}^{\len(B_d)}P_{\ell_k(B_d), s_k(B_d), v_k(B_d)}^\lambda,
    \end{equation}
    where $P_{\ell_k(B_d), s_k(B_d), v_k(B_d)}^\lambda$ denotes the single resonator propagation matrix as defined in \cref{eq:propmat}, but with the resonator parameters $\ell_k(B_d), s_k(B_d), v_k(B_d)$. Note that $\prod_{i=1}^k A_i := A_k A_{k-1} \ldots A_1$ for matrices $A_1, \ldots, A_k$.

    Given $\chi \in \ldz$,  define the associated \emph{block propagation matrix cocycle} 
    \begin{equation}
        \begin{aligned}
            \mc P^\lambda:\Z&\to \on{SL}(2,\R)\\
            j &\mapsto \mc P^\lambda_{\seq{\chi}{j}}.
        \end{aligned}
    \end{equation}
\end{definition}
By construction, $\mc P^\lambda$ takes at most $D$ distinct values arranged according to the block sequence $\chi\in \ldz$.
To avoid confusion, we will also refer to the propagation matrix cocycle $P^\lambda$ defined in the previous subsection as the \emph{resonator propagation matrix cocycle}.

We now find that the uniform hyperbolicity of the block propagation matrix cocycle $\mc P^\lambda$ is equivalent to that of the resonator propagation matrix cocycle $P^\lambda$.
\begin{lemma}\label{lem:UHiffUH_contracted}
    Consider a block sequence $\chi \in \ldz$ with associated resonator sequence $\alpha = \Phi_0(\chi) \in \lrz$ and denote by $\mc P^\lambda$ and $P^\lambda$ the block propagation matrix and resonator propagation matrix cocycles, respectively. 

    Then, for any $\lambda\geq 0$, we have
    \[
        \mc P^\lambda\in \mc{UH} \iff  P^\lambda\in \mc{UH}.
    \]
\end{lemma}
\begin{proof}
Analogously to $\Phi_0:\ldz \to \lrz$, for any block sequence $\chi\in \ldz$ 
we define the map $\Phi_\chi':\Z\to \Z$ taking any block index $j\in \Z$ to the index $i$ of the first resonator corresponding to that block. In particular, $\Phi_\chi'$ is recursively defined as 
\[
\Phi_\chi'(0) = 0, \quad \Phi_\chi'(j+1) = \Phi'_\chi(j)+\len(B_{\oo{j}}).
\]

\enquote{$\impliedby$}:
Consider a resonator sequence $\alpha = \Phi_0(\chi)\in \lrz$ and $\lambda\geq0$ such that $P^\lambda$ is uniformly hyperbolic with directions $u,s:\Z\to \rp$. We can then define the unstable and stable block directions $q, p:\Z \to \rp$ as $q(j) = u(\Phi'_\chi(j))$ and $p(j) = s(\Phi'_\chi(j))$. By construction, these sequences satisfy $$\mc P^\lambda(j)q(j) = q(j+1), \mc P^\lambda(j)p(j) = p(j+1),$$ as well as the uniform growth condition (with the same constants, even though they might potentially be sharpened)

\enquote{$\implies$}:
Consider a resonator sequence $\alpha = \Phi_0(\chi)\in \lrz$ and $\lambda\geq0$ such that $\mc P^\lambda$ is uniformly hyperbolic with directions $u,s:\Z\to \rp$. We define the unstable and stable resonator directions $q, p:\Z \to \rp$ as $q(i) = P^\lambda_i(0)u(0)$ and $p(i) = P^\lambda_i(0)s(0)$. By definition, these sequences are $P^\lambda$-invariant. 

It remains to demonstrate their uniform growth. By definition, we have $q(\Phi'_\chi(j)) = u(j)$ and $p(\Phi'_\chi(j)) = s(j)$ and thus the sequences $q$ and $p$ satisfy the uniform growth condition at least at the points $\Phi'_\chi(j)$ for $j\in \Z$. 

But, because $P^\lambda\in \ell^\infty(\Z, \operatorname{SL}(2,\R))$, we know that $\norm{P^\lambda(i)}$ and $\norm{(P^\lambda(i))^{-1}}$ must be uniformly bounded, which together with the fact that there exists a maximal block length $\len_{max} = \max_{d=1,\dots, D} \len (B_d)$ ensures that 
$\norm{P^\lambda_n(i)\vec{q}(i)}$ and $\norm{P^\lambda_{-n}(i)\vec{p}(i)}$ cannot grow arbitrarily between the points of known decay $\Phi'_\chi(j)$. This allows us to find constants $C>0$ and $\eta>1$ such that $q$ and $p$ satisfy the uniform growth condition.
\end{proof}

Finally, due to the regularity of the block propagation matrix cocycle, we can apply the invariant cone criterion of uniform hyperbolicity.
\begin{lemma}\label{lem:UHiffCone}
    Let $\chi\in \ldz$ be a pseudo-ergodic block sequence and $\lambda\geq 0$. Then, the associated block propagation matrix cocycle $\mc P^\lambda$ is uniformly hyperbolic if and only if there exists a nonempty open subset $M\subset \rp$ with $\overline{M}\neq \rp$ such that 
    \[
        {\mc P_d^\lambda(\overline{M})}\subset M 
    \]
    for all $d=1,\dots,D$.
\end{lemma}
\begin{proof}
    This is essentially \cite[Theorem 2.2]{avila.bochi.ea2010Uniformly} where instead of considering uniform hyperbolicity for dynamically defined cocycles over the full shift, we consider it for a single pseudo-ergodic sequence. This works because the pseudo-ergodicity of $\chi$ ensures the full shift over the generating set. We refer to Appendix \ref{sec:dyndefined} for a detailed discussion of these equivalent formulations.
\end{proof}
We shall call such a set $M$ a \emph{strictly invariant multicone}. Notably, the existence of such an invariant cone depends only on the properties of $\mc P_d^\lambda$ for $d=1,\dots,D,$ and does not depend on the block sequence $\chi$.

We now give a concrete sufficient condition for the existence of such an invariant cone. To that end, we shall call the block propagation matrix $\mc P_d^\lambda$ \emph{hyperbolic} if 
\begin{equation}
    \abs{\tr \mc P_d^\lambda}>2,
\end{equation}
where $\tr$ denotes the trace.
\begin{remark}\label{rmk:notUHifnotH}

    Suppose that there is $d\in\{1,\ldots,D\}$ such that $|\operatorname{tr}\mc{P}_d^\lambda|\le2$, and set $A=\mc{P}_d^\lambda$. If $|\operatorname{tr}A|<2$, then $A$ is elliptic and conjugate to a rotation, so $\|A^n\|\le C$. If $|\operatorname{tr}A|=2$, then $A=\pm I$ or $A$ is parabolic and conjugate to
\[
 \pm\begin{pmatrix}1&c\\0&1\end{pmatrix};
\]
in either case $\|A^n\|\le C(1+n)$. Hence,
\begin{equation}
 \|(\mc{P}_d^\lambda)^n\|\le C(1+n).
\label{eq:nonhyperbolic-growth}
\end{equation}
Since $\chi$ is pseudo-ergodic, for every $n$ it contains a constant sequence of length $n$ consisting only of $d$. Thus, there is $j_n$ such that
\[
 \bigl(\mc{P}^\lambda\bigr)_n(j_n)=(\mc{P}_d^\lambda)^n.
\]
Uniform hyperbolicity would imply constants $c>0$ and $\xi>1$ such that
\[
 \|\bigl(\mc{P}^\lambda\bigr)_n(j)\|\ge c\xi^n
\]
for all $j$ and $n$, contradicting~\eqref{eq:nonhyperbolic-growth}. Therefore, the block cocycle cannot be uniformly hyperbolic.


\end{remark}

Recall that because $\mc P_d^\lambda\in \operatorname{SL}(2,\R)$, we can identify the block propagation matrices with M\"obius transformations on the real projective space $\rp$, which in turn is diffeomorphic to the sphere $\rp\simeq \mathbb{S}^1$. In this language, $\mc P_d^\lambda$ is hyperbolic if and only if the corresponding M\"obius transformation is hyperbolic, which means that it is represented by a matrix such that the modulus of its trace is strictly larger than two or equivalently, it has two distinct fixed points, one source and one sink, which we shall denote by $s(\mc P_d^\lambda)\in \rp$ and $u(\mc P_d^\lambda) \in \rp$, respectively
\footnote{Namely, let $(\xi_1, E_1), (\xi_2, E_2)$ be the eigenvalues and eigenspaces of $\mc P_d^\lambda$ such that $\abs{\xi_1}<1<\abs{\xi_2}$. After identifying these dimension $1$ eigenspaces with points in $\rp$, we find that $s(\mc P_d^\lambda)=E_1$ and $u(\mc P_d^\lambda)=E_2$.}.

\begin{figure}
    \begin{subfigure}{0.4\textwidth}
        \centering
        \includegraphics[width=0.5\textwidth]{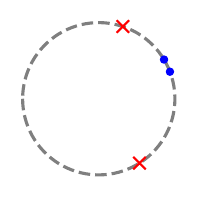}
        \caption{Source-sink condition satisfied}
    \end{subfigure}
    \begin{subfigure}{0.4\textwidth}
        \centering
        \includegraphics[width=0.5\textwidth]{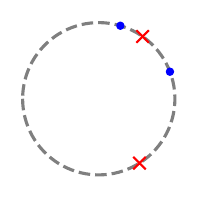}
        \caption{Source-sink condition violated}
    \end{subfigure}
    \caption{Illustration of the spectral consequences of violating the source-sink condition in $\rp\simeq S^1\subset \R^2$ for a family of two block propagation matrices. Sources are marked as red crosses and sinks as blue points.}
    \label{fig:source_sink_illustration}
\end{figure}

Now, we are in a position to give a sufficient condition ensuring the existence of an invariant cone for a family of hyperbolic block propagation matrices.
\begin{definition}[Source-sink condition]
    We say that the family of hyperbolic block propagation matrices $\{\mc P_1^\lambda, \dots, \mc P_D^\lambda\}$ satisfies the \emph{source-sink condition} if all the sink fixed points $u(P_d^\lambda),$ for $ d=1,\dots,D,$ lie in the same connected component of $\rp\setminus\{s(P_1^\lambda), \dots, s(P_D^\lambda)\}$.
\end{definition}
This definition is illustrated in \cref{fig:source_sink_illustration} for a family of $D=2$ block propagation matrices.

\begin{corollary}\label{cor:invarcone}
    Let $\lambda\geq0$ and assume that $\mc P_d^\lambda$ is hyperbolic for every $d=1,\dots,D,$ and that the family $\{\mc P_1^\lambda, \dots, \mc P_D^\lambda\}$ satisfies the source-sink condition. Then, the family has an invariant cone.
\end{corollary}
\begin{proof}
    For every $d=1,\dots, D,$ let $\xi^s_d\in \R, s_d:=s(P^\lambda_d)\in \rp$ and $\xi^u_d \in \R, u_d:= u(P^\lambda_d) \in \rp$ denote the source and the sink eigenvalue, respectively, and the corresponding fixed point of $\mc P^\lambda_d$. Define
    \begin{align*}
        \phi_d:\rp \setminus \{s_d\} &\to \R\\
        v = \alpha\vec{u}_d + \beta\vec{s}_d & \mapsto \frac{\beta}{\alpha},
    \end{align*}
    where $\vec{u}_d$ and $\vec{s}_d$ are unit–length vectors in the sink and source eigenspaces, respectively. It is easy to see that $\phi_d$ is a homeomorphism for all $d$.

    We have $\phi_d(u_d) = 0$ and 
    \[
        \phi_d(\mc P^\lambda_d v) = \phi_d(\mc P^\lambda_d(\alpha\vec{u}_d + \beta\vec{s}_d)) = \phi_d(\xi_d^u\alpha\vec{u}_d + \xi_d^s\beta\vec{s}_d) = \frac{\xi_d^s}{\xi_d^u} \phi_d(v)\eqqcolon \gamma_d\phi_d(v),
    \]
    for some $v\in \rp\setminus \{s(\mc P^\lambda_d)\}$ and $0<{\gamma_d}<1$.

    By assumption, all sinks $u_1, \dots, u_D,$ lie in the same connected component $C\subset \rp \setminus S$ for $S \coloneqq \{s_d \mid d=1,\dots, D\}$. We can thus find a nonempty open connected set $M\subset C$ such that $\overline{M}\cap S = \emptyset$ and $\{u_d\mid d =1 ,\dots,D\}\subset M$.

    For every $d=1,\dots, D,$ there exist $a_d<0<b_d$ such that $\phi_d(M)=(a_d, b_d)$. In particular, we must have $0 \in \phi_d(M)$ for all $d=1,\dots,D,$ and hence $a_d<0<b_d$.
    We then have \[
        \phi_d(\mc P^\lambda_d (\overline{M})) = \gamma_d \phi_d(M) =  [\gamma_d a_d, \gamma_d b_d] \subset (a_d,b_d)= \phi_d(M),
    \]
    and thus ${\phi_d(\mc P^\lambda_d (\overline{M}))} \subsetneq (a_d,b_d)=\phi_d(M)$. Since $\phi_d$ is a homeomorphism, this implies ${\mc P^\lambda_d (\overline{M})}\subsetneq  M$, which completes the proof.
\end{proof}

We are now in a position to give explicit sufficient conditions for a spectral gap in the spectrum $\mc J(\chi)$ in terms of the corresponding block propagation matrices $\{\mc P_1^\lambda, \dots, \mc P_D^\lambda\}$,
together with a partial converse obstruction.

\begin{theorem}\label{thm:saxonhutner}
    Consider an infinite block disordered system with blocks $B_1, \dots, B_D$, pseudo-ergodic block sequence $\chi \in \ldz$ and corresponding Jacobi operator $\mc J(\chi)$. 
    
    Let $\lambda\geq 0$ and assume that all block propagation matrices $\mc P_d^\lambda$ are hyperbolic, \emph{i.e.},
    \[
        \abs{\tr \mc P_d^\lambda} > 2 \quad \text{ for all }d=1,\dots, D.
    \]
    Assume further that the family $\{\mc P_1^\lambda, \dots, \mc P_D^\lambda\}$ satisfies the \emph{source-sink condition}. Then, we must have 
    \[
        \lambda\notin \sigma(\mc J(\chi)).
    \]

    Conversely, if we have $\abs{\tr \mc P_d^\lambda} \leq 2$ for any $d=1,\dots D$, then $\lambda \in \sigma(\mc J(\chi))$.
\end{theorem}
\begin{proof}
    By combining \cref{thm:johnson} together with \cref{lem:DSiffUH,lem:DSiffDS_cohomology,lem:UHiffUH_contracted}, we find that $\lambda\notin \sigma(\mc J(\chi))$ if and only if the block propagation matrix cocycle $j\mapsto \mc P^\lambda(j) = \mc P^\lambda_{\oo{j}}$ is uniformly hyperbolic. \cref{lem:UHiffCone,cor:invarcone} then ensure uniform hyperbolicity if the source-sink condition is fulfilled. 
    
    Conversely, by \cref{rmk:notUHifnotH}, we know that any block propagation matrix that is not hyperbolic causes the entire block propagation matrix to not be uniformly hyperbolic, which implies that $\lambda\in \sigma(\mc J(\chi))$.
\end{proof}

The above proof structure can thus be illustrated as follows:
\begin{equation}\label{eq:saxon-hutner_schematic}
    \begin{tikzcd}
	{\lambda\notin\sigma(\mathcal{J}(\chi))} & {T^\lambda\in \mathcal{DS}} \\
	& {P^\lambda\in\mathcal{DS}} \\ 
    &{P^\lambda\in\mathcal{UH}} \\
	& {\mathcal{P}^\lambda\in\mathcal{UH}} & {\{\mc P_1^\lambda, \dots, \mc P_D^\lambda\} \text{\scriptsize \; has an invariant cone}}\\
	& \mc P_1^\lambda, \dots, \mc P_D^\lambda \text{\scriptsize \; hyperbolic} & \mc P_1^\lambda, \dots, \mc P_D^\lambda \text{\scriptsize \; hyperbolic + source-sink condition}
	\arrow["\cref{thm:johnson}",Leftrightarrow, from=1-1, to=1-2]
	\arrow["\cref{lem:DSiffDS_cohomology}", Leftrightarrow, from=1-2, to=2-2]
	\arrow["\cref{lem:DSiffUH}",Leftrightarrow, from=2-2, to=3-2]
	\arrow["\cref{lem:UHiffUH_contracted}",Leftrightarrow, from=3-2, to=4-2]
	\arrow["\cref{lem:UHiffCone}",Leftrightarrow, from=4-2, to=4-3]
	\arrow["\cref{rmk:notUHifnotH}",Rightarrow, from=4-2, to=5-2]
	\arrow["\cref{cor:invarcone}",Leftarrow, from=4-3, to=5-3]
    \end{tikzcd}
\end{equation}

\begin{remark}
    Under certain symmetry assumptions, it is possible to show that the conditions of \cref{thm:saxonhutner} reduce to conditions for the eigenvectors at the band edges of the blocks, which are much easier to check. This and other effects of symmetry in block disordered systems will be the subject of upcoming work.
\end{remark}
The monomer and dimer blocks in \cref{ex:standard_blocks} are specifically chosen such that this source-sink condition is satisfied for all $\lambda$ in the shared bandgap. Therefore, for any pseudo-ergodic sequence of these two blocks, \cref{thm:saxonhutner} ensures that a value $\lambda$ lies in the bandgap of the total system if and only if it lies in the bandgap of both blocks. As can be seen in \cref{fig:saxonhutner_edgemode}(\textsc{a}), this characterisation continues to hold even in the finite case. The justification of this observation will be the subject of the following section.

\begin{figure}
    \centering
    \includegraphics[width=0.5\linewidth]{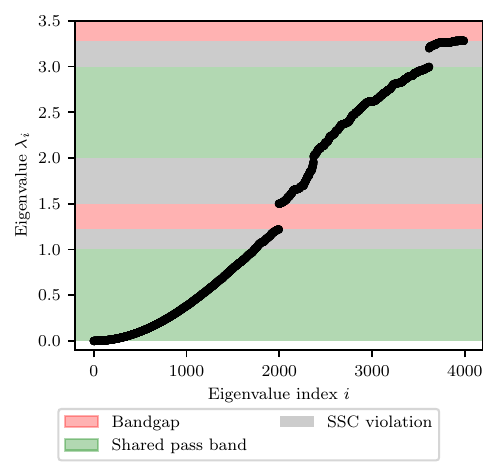}
    \caption{Illustration of the necessity of the source-sink condition (SSC) for a block disordered system with dimer and inverted dimer block sampled $M=2000$ times. Regions in the shared bandgap where the SSC is violated are marked in gray. The SSC violation leads to spectral pollution in these regions.}
    \label{fig:SSC_violation}
\end{figure}
Finally, we aim to briefly discuss the violation of the source-sink condition. To that end, we consider two dimer blocks: one as in \cref{ex:standard_blocks} and the other with swapped spacings (\emph{i.e.} $s_1=2$, $s_2=1$). While from their symmetry it immediately follows that these blocks have the same band structure, their combination violates the SSC in large parts of the shared bandgaps. In \cref{fig:SSC_violation}, we plot a block disordered system consisting of these blocks and mark the parts of the shared bandgap where the SSC is violated. As can be seen, this leads to spectral pollution in these areas due to what are essentially interface modes between the two distinct types of blocks. However, we would also like to clarify that, while this spectral pollution is widespread, it does not completely cover the entire SSC violation region, preventing the SSC from being a strictly necessary condition. This is due to the potential existence of more complex invariant cones (see \cite{avila.bochi.ea2010Uniformly}) -- even when the SSC is violated.

\section{Semi-infinite and finite Jacobi operators}\label{sec:semiinf_and_finite}
Having characterised the spectrum $\sigma(\mc J(\chi))$ of the bi-infinite Jacobi operator determined by a bi-infinite block sequence $\chi\in\ldz$, we now aim to do so for the semi-infinite and finite case. 

To begin, we must specify the appropriate matrices and operators. For the finite case, following \cref{thm:capapprox}, the appropriate spectral problem is the capacitance matrix problem $(VC-\lambda)\bm u = 0$. To connect the capacitance matrix formulation to the bi-infinite Jacobi operators, we instead consider the similar $VC \sim V^{\frac{1}{2}}CV^{\frac{1}{2}}$, and will find that it only differs at the edges. For the semi-infinite case, we then enforce the same edge values as in the finite case.
\begin{definition}
 Given a finite sequence $\chi_M\in D^M$, we denote the \emph{finite Jacobi matrix} $\mc J_N(\chi_M) \coloneqq V^{\frac{1}{2}}CV^{\frac{1}{2}}\in \R^{N\times N}$. Here, $N\in \N$ corresponds to the length $N$ of the associated resonator sequence $\alpha \in \mathbb{L}_R(N)$, determined by the block sequence $\chi_M$. From the definitions \cref{eq:matmat,eq:capmat} of $V$ and $C$, it follows that the matrix $\mc J_N(\chi_M) = V^{\frac{1}{2}}CV^{\frac{1}{2}}$ is again tridiagonal and symmetric with bands as in \cref{eq:jacobibands}, except for the fact that the upper left and lower right entries are perturbed:
\begin{equation}\label{eq:physicalbv}
     \bseq{b}{0} = \frac{v_0^2}{\ell_0}\left(\frac{1}{s_0}\right) \quad \text{and} \quad 
        \bseq{b}{N-1} = \frac{v_{N-1}^2}{\ell_{N-1}}\left(\frac{1}{s_{N-2}}\right),
\end{equation}
where the sequences $(v_i)_{i=0}^{N-1}, (s_i)_{i=0}^{N-1}, (\ell_i)_{i=0}^{N-1}$ are determined by the resonator sequence $\alpha\in \mathbb{L}_R(N)$.

Similarly, for a semi-infinite block sequence $\chi_+\in \ldzp$, we define the \emph{positive semi-infinite Jacobi operator} as
\begin{align*}
    \mc J_+(\chi_+):\ell^2(\Z_{\geq 0}) &\to \ell^2(\Z_{\geq 0})\\
    \bseq{v}{i} &\mapsto (J\bm v)^{(i)} = \begin{cases}
        \bseq{a}{i-1}\bseq{v}{i-1} + \bseq{a}{i}\bseq{v}{i+1} + \bseq{b}{i}\bseq{v}{i} & i>0,\\
        \bseq{a}{i}\bseq{v}{i+1} + \bseq{b}{i}\bseq{v}{i} & i=0,\\
    \end{cases}
\end{align*}
where again $\bseq{b}{i}$ and $\bseq{a}{i}$ as in \cref{eq:jacobibands} except for at the edge $\bseq{b}{0} = v_0^2/(\ell_0s_0)$. The \emph{negative semi-infinite Jacobi operator} $\mc J_-(\chi_+):\ell^2(\Z_{\leq 0}) \to \ell^2(\Z_{\leq 0})$ is defined analogously with edge $\bseq{b}{0} = v_0^2/(\ell_0s_{-1})$.
\end{definition}

\subsection{Finite section method}

To connect the spectrum of a bi-infinite Jacobi operator to that of its semi-infinite counterpart, we use the limit-operator characterisation of Fredholmness for band operators, specifically \cite[Theorem 3.2]{lindner2012finite}, together with the half-axis compression result \cite[Lemma 3.1]{lindner2012finite}. The pseudo-ergodic identification of limit operators in \cite[Theorem 2.7]{lindner2012finite} must be adapted to the present block-constrained family: for a pseudo-ergodic block sequence, the limit operators are exactly the bi-infinite Jacobi operators generated by the full block hull. Consequently, the essential spectrum of any pseudo-ergodic half-line operator equals the union of the spectra of all bi-infinite operators in this hull.

\begin{proposition}\label{prop: big_equivalence}
    Let $\lambda\geq 0$. The following statements are equivalent:
    \begin{enumerate}[label=(\roman*)]
        \item There exists a pseudo-ergodic $\chi \in \ldz$ such that $\mc J(\chi)-\lambda$ is Fredholm;
        \item $\mc J(\chi)-\lambda$ is Fredholm for all $\chi \in D^{\Z}$;
        \item $\mc J(\chi)-\lambda$ is invertible for all $\chi \in D^{\Z}$;
        \item there exists a pseudo-ergodic $\chi_+ \in D^{\Z_{\geq0}}$ such that $\mc J_\pm(\chi_+)-\lambda$ is Fredholm;
        \item $\mc J_\pm(\chi_+)-\lambda$ is Fredholm for all $\chi_+ \in D^{\Z_{\geq0}}$.
    \end{enumerate}
\end{proposition}

We note that, due to \cref{thm:johnson} and the arguments of the previous section, all of these statements are further equivalent to 
\begin{enumerate}
    \item[(vi)]  There exists a pseudo-ergodic $\chi \in \ldz$ such that the associated block propagation matrix cocycle $\mc P^\lambda:\Z\to \on{SL}(2,\R)$ is uniformly hyperbolic.
\end{enumerate}

We can rephrase the above result as follows. 
\begin{corollary}\label{lem:weylinclusions}
Let $\chi_+\in D^{\Z_{\geq 0}}$ be pseudo-ergodic. Then, for any $\chi'\in \ldz$, it holds that
\[
    \sigma(\mc J(\chi'))\subset \sigma(\mc J_\pm(\chi_+)).
\]

If, in addition, $\chi'\in \ldz$ is pseudo-ergodic, then we have\footnote{Here, $\sigma_{ess}(A)$ denotes the \emph{essential spectrum} for some self-adjoint operator $A$, defined as the spectrum of $A$ without the isolated eigenvalues of finite multiplicity.}
\[
    \sigma_{ess}(\mc J(\chi')) = \sigma(\mc J(\chi')) = \sigma_{ess}(\mc J_\pm(\chi_+)).
\]
\end{corollary}

Thus, the spectrum of the semi-infinite $\mc J_+(\chi_+)$ contains at least the spectrum of the bi-infinite operator. 

\subsection{Edge modes for semi-infinite Jacobi operators}
By the previous subsection, to fully describe the spectrum $\sigma(\mc J_+(\chi_+))$ it remains to characterise eigenvalues $\lambda \in \sigma(\mc J_+(\chi_+)) \setminus \sigma(\mc J(\chi'))$ for some $\chi'\in \ldz$ such that the semi-infinite truncation of $\chi'$ equals $\chi_+$, \emph{i.e.}, $(\chi')^{(i)} = \chi_+^{(i)}$ for all $i\in \Z_{\geq0}$. We will denote this by $\chi_+\prec \chi'$.

The crucial difference between $\mc J(\chi)$ and $\mc J_+(\chi_+)$ is that $\mc J_+(\chi_+)$ exhibits an edge at its left end, breaking translation invariance. This allows for the existence of \emph{edge-modes} exponentially localised at this edge. 
 By the following result, the additional modes in $\sigma(\mc J_+(\chi_+)) \setminus \sigma(\mc J(\chi'))$ are exactly such edge modes, which, due to the physical origin of the propagation matrices, can be explicitly characterised as eigenvalues of multiplicity one. They occur exactly when the stable, decaying direction at the left end of the array $s(0)$ exactly matches the Neumann boundary condition $(1,0)^\top$.


\begin{theorem}\label{thm:semiinfinite-johnson}
    Let $\chi_+\in D^{\Z_{\geq0}}$ be pseudo-ergodic and let $\chi'\in \ldz$ be such that $\chi_+\prec \chi'$. If $\lambda\in \sigma(\mc J_+(\chi_+)) \setminus \sigma(\mc J(\chi'))$, then the block propagation matrix cocycle $\mc P^\lambda:\Z\to \on{SL}(2,\R)$ corresponding to the bi-infinite sequence $\chi'$ is uniformly hyperbolic and $(1,0)^\top\in s(0)$, where $s$ denotes the stable direction of $\mc P^\lambda$.

    Furthermore, in this case $\lambda \in \sigma(\mc J_+(\chi_+))$ is an eigenvalue of multiplicity one, with the eigenvector $\bm v\in \ell^2(\Z_{\geq0})$ explicitly given by 
    \[
        \begin{pmatrix}
            \bseq{v}{i}\\
            \bseq{v}{i-1}
        \end{pmatrix} = Q^{-1}(i)P^\lambda_i(0)\begin{pmatrix}
            1\\0
        \end{pmatrix}.
    \]

    Conversely, consider any $\lambda\notin \sigma(\mc J(\chi'))$ such that $(1,0)^\top\in s(0)$. Then $\lambda\in \sigma(\mc J_+(\chi_+))$ is an eigenvalue with multiplicity one and the corresponding eigenvector as above.
\end{theorem}
\begin{proof}
    From $\lambda\notin \sigma(\mc J(\chi'))$, we immediately get from \cref{thm:johnson} together with \cref{lem:DSiffUH,lem:DSiffDS_cohomology,lem:UHiffUH_contracted} that the block propagation matrix cocycle $\mc P^\lambda:\Z\to \on{SL}(2,\R)$ corresponding to the bi-infinite sequence $\chi'$ must be uniformly hyperbolic.

    For the second part, we first note that, because $\chi_+$ is pseudo-ergodic and $\chi_+\prec\chi'$ we also have $\chi'$ pseudo-ergodic. Consequently, by \cref{lem:weylinclusions} we have $\lambda\notin \sigma(\mc J(\chi'))$ implies $\lambda\notin \sigma_{ess}(\mc J_+(\chi_+))$ and, therefore, there must exist a $\bm v \in \ell^2(\Z_{\geq 0})$ such that $(\mc J_+(\chi_+) - \lambda)\bm v = 0$.

    Recall that we chose $\mc J_+(\chi_+)$ to have the edge values \cref{eq:physicalbv} and thus get the equations
    \[
        \bseq{a}{0}\bseq{v}{1} + (\bseq{b}{0}-\lambda)\bseq{v}{0} = 0
    \]
    and
    \[
        \bseq{a}{i-1}\bseq{v}{i-1} + \bseq{a}{i}\bseq{v}{i+1} + (\bseq{b}{i}-\lambda)\bseq{v}{i} = 0 \quad \text{for all }i\in \N.
    \]

    Without loss of generality, we may choose $\bseq{v}{0} = c\in \R$ and thus find $$(\bseq{v}{1}, \bseq{v}{0})^\top = c(-\frac{\bseq{b}{0}-\lambda}{\bseq{a}{0}},1)^\top$$ from the first equation and 
    \[
        \begin{pmatrix}
            \bseq{v}{i+1}\\
            \bseq{v}{i}
        \end{pmatrix} = cT_i^\lambda(1)\begin{pmatrix}
            -\frac{\bseq{b}{0}-\lambda}{\bseq{a}{0}}\\
            1
        \end{pmatrix}
    \]
    from the latter equations. Here, $T^\lambda_i(1)$ denotes the transfer matrix cocycle associated with $\mc J(\chi')$ as defined in \cref{eq:transfermat_cocycle_iteration}. Notably, even though $\bm v$ is an eigenvector of $\mc J_+(\chi_+)$, it is still determined by the transfer matrix cocycles of $\chi'$ as the physical boundary conditions leave $\bseq{a}{i}$ and $\bseq{b}{i+1}$ unchanged for $i\in \Z_{\geq 0}$. In particular, the above characterisation also ensures that the $\lambda$-eigenspace of $\mc J_+(\chi_+)$ is one-dimensional, as it completely determines the eigenvector up to the factor $c$.

    It remains to connect this characterisation to the resonator propagation matrix cocycle. To that end, in line with \cref{rmk:propmat_intuition}, we define
    \[
        \begin{pmatrix}
            u(x_i^{\iL})\\
            u'(x_i^{\iL})
        \end{pmatrix} \coloneqq Q(i) \begin{pmatrix}
            \bseq{v}{i}\\
            \bseq{v}{i-1}
        \end{pmatrix} \quad \text{for }i\in \N,
    \]
    and consequently have 
    \[
        \begin{pmatrix}
            u(x_{i+1}^{\iL})\\
            u'(x_{i+1}^{\iL})
        \end{pmatrix} = P^\lambda(i) \begin{pmatrix}
            u(x_i^{\iL})\\
            u'(x_i^{\iL}) 
        \end{pmatrix}\quad \text{for }i\in \N,
    \]
    where $P^\lambda(i)$ is the resonator propagation matrix cocycle associated with $\mc J(\chi')$, introduced in \cref{def:propmat}. Now, we may use $P^\lambda(0)$ to define $(u(x_{0}^{\iL}), u'(x_{0}^{\iL}))^\top$ and find
    \[
        \begin{pmatrix}
            u(x_{0}^{\iL})\\
            u'(x_{0}^{\iL})
        \end{pmatrix} \coloneqq (P^\lambda(0))^{-1}\begin{pmatrix}
            u(x_{1}^{\iL})\\
            u'(x_{1}^{\iL})
        \end{pmatrix} = c(P^\lambda(0))^{-1}Q(1) \begin{pmatrix}
            -\frac{\bseq{b}{0}-\lambda}{\bseq{a}{0}}\\
            1
        \end{pmatrix} = c'\begin{pmatrix}
            1\\
            0
        \end{pmatrix},
    \]
    where the last equality follows after multiplying out the respective matrices.

    Finally, we note that because the family of matrices $Q(i)$ only takes finitely many distinct values, we have $\bm v\in \ell^2(\Z_{\geq 0})$ if and only if $(\norm{(u(x_i^{\iL}), u'(x_i^{\iL}))^\top})_{i=0}^{\infty}\in \ell^2(\Z_{\geq 0})$. By \cref{lem:UHiffUH_contracted}, we know that the resonator propagation matrix cocycle $P^\lambda$ is uniformly hyperbolic if the block propagation matrix cocycle $\mc P^\lambda$ is and from the proof of that lemma, it follows that their respective stable directions both take the same value at zero, $s(0)\in \rp$. Consequently, $\bm v$ can only be $\ell^2$-localised and thus an eigenvector if $(u(x_0^{\iL}), u'(x_0^{\iL}))^\top = c'(1,0)^\top \in s(0)$, as desired.

    The converse direction follows immediately from the characterisation given above.
\end{proof}
There is a straightforward criterion that prevents the existence of edge modes closely related to the source-sink condition. We again use $s(\mc P_d^\lambda)\in \rp$ and $u(\mc P_d^\lambda) \in \rp$ to denote the stable and unstable block propagation matrix fixed points, respectively. 
\begin{corollary}\label{cor:edgemode_characterisation}
    Suppose that the family of block propagation matrices $\mc P_1^\lambda, \dots, \mc P_D^\lambda,$ satisfies the assumptions of \cref{cor:invarcone}.
    Then, if $(1,0)^\top$ lies in the maximal connected component of $\rp\setminus \{s(\mc P_1^\lambda), \dots s(\mc P_D^\lambda)\}$ containing all $u(\mc P_1^\lambda), \dots, u(\mc P_D^\lambda)$, we cannot have $(1,0)^\top \in s(0)$ for any associated block propagation matrix cocycle determined by a block sequence $\chi \in \ldz$.
\end{corollary}

\begin{proof}
Consider hyperbolic block propagation matrices $\{\mc P^\lambda_1,\dots \mc P^\lambda_D\}$ satisfying the source-sink condition and a sequence $\chi\in \ldz$ inducing a block propagation matrix cocycle $\mc P^\lambda(j), j\in \Z$. By \cref{cor:invarcone}, we know that this cocycle is uniformly hyperbolic and we denote $s(j), u(j)\in \rp, j\in \Z$ the stable and unstable directions, respectively.


Let $U$ be the connected component of
\[
 \rp \setminus\{s(\mc P_1^\lambda),\ldots,s(\mc P_D^\lambda)\}
\]
containing all $u(\mc P_1^\lambda),\ldots,u(\mc P_D^\lambda)$, and set $S:=\rp \setminus U$.

We claim that $s(j)\in S$ for every $j$.  Suppose instead that $s(j')\in U$.  Choose a connected open interval $M$ such that
\[
 \{s(j'),u(\mc P_1^\lambda),\ldots,u(\mc P_D^\lambda)\}\subset M,
 \qquad \overline M\subset U.
\]
As in the proof of \Cref{cor:invarcone}, let $\phi_d$ be the projective eigen-coordinate for $\mc P_d^\lambda$.  Then
\[
 \phi_d(\mc P_d^\lambda v)=\gamma_d \phi_d(v),
 \qquad
 \gamma_d:=\frac{\xi_d^s}{\xi_d^u}\in(0,1).
\]

Since $\phi_d(M)=(a_d,b_d)$ with $a_d<0<b_d$, we have
\[
 \phi_d(\mc P_d^\lambda(\overline M))
 =\gamma_d[a_d,b_d]\subset(a_d,b_d) =\phi_d(M),
\]
and hence,
\begin{equation}
 \mc P_d^\lambda(\overline M)\Subset M.
\label{eq:strict-cone}
\end{equation}

Equip $M$ with its Hilbert projective metric $d_M$.  In an affine coordinate in which $M=(a,b)$, for $a<x<y<b$ it is
\[
 d_M(x,y)=\log\frac{(y-a)(b-x)}{(x-a)(b-y)}, \quad a <x<y<b.
\]
By~\eqref{eq:strict-cone} and the Birkhoff contraction Theorem \cite{birkhoff1957extensions}, there is a common $\tau\in(0,1)$ such that
\begin{equation}
 d_M(P_d^\lambda v_1,\mc P_d^\lambda v_2)
 \le\tau d_M(v_1,v_2)
\label{eq:hilbert-contraction}
\end{equation}
for every $d$ and $v_1,v_2\in M$.

Strict invariance also gives $u(j)\in M$ for all $j$, through the nested-cone construction
\[
 \{u(j)\}
 =\bigcap_{n\ge1} \mc P^\lambda(j-1)\cdots \mc P^\lambda(j-n)\overline M.
\]
Since both $s(j')$ and $u(j')$ belong to $M$, invariance and the contraction \eqref{eq:hilbert-contraction} imply that
\[
 d_M(s(j'+n),u(j'+n))
 \le\tau^n d_M(s(j'),u(j'))\longrightarrow0.
\]
The Hilbert metric induces the usual topology on $M$, so the angular distance also tends to zero. This contradicts the uniform separation of the stable and unstable directions.  Hence, $s(j)\in S$ for all $j$.

Now, if $(1,0)^\top\in U$, then we can find an open invariant cone such that $(1,0)^\top\in M$ and $\overline{M}\subset U$. By the above claim, we have $s(j)\in S$ for all $j\in \Z$ and $S\cap M = \emptyset$ ensuring, in particular, $s(0) \neq (1,0)^\top$. Hence, the proof of the corollary is complete.
\end{proof}



These arguments can be carried out completely analogously for the negative semi-infinite Jacobi operators $\mc J_-(\chi_+)$.

Finally, in this subsection, we state what can essentially be regarded as Coburn's lemma for semi-infinite resonator systems. This result is not as strong as the \emph{stochastic Coburn's lemma} obtained in \cite{chandler-wilde.lindner2016Coburns}, which is due to the relaxed notion of pseudo-ergodicity. Nonetheless, it is still able to enforce the following condition on edge modes: depending on the material parameters, at any frequency, disordered systems might support an eigenmode supported at the left or right edge of the system, but never at both.
\begin{theorem}\label{thm:coburn}
    Suppose that there exists some pseudo-ergodic $\chi' \in \ldz$ such that $\mc J(\chi') - \lambda$ is Fredholm. Then, for every $\chi_+ \in \ldzp$, either $\mc J_+(\chi_+)-\lambda$ or $\mc J_-(\chi_+)-\lambda$ is invertible.
\end{theorem}
\begin{proof}
    Fredholmness of $\mc J_\pm(\chi_+)-\lambda$ is an immediate consequence of Proposition \ref{prop: big_equivalence}. 
    
    By contradiction, suppose that there exists a sequence $\chi_+ \in \ldzp$ such that $\mc J_+(\chi_+) - \lambda$ and $\mc J_-(\chi_+) -\lambda$ both have nontrivial kernels. For the rest of this proof, we shall use the convention $\Z_+\coloneqq \Z_{\geq0}$ and $\Z_-\coloneqq \Z_{<0}$ as well as $\mc J_+(\chi_+):\ell^2(\Z_+)\to \ell^2(\Z_+)$ and $\mc J_-(\chi_+):\ell^2(\Z_-)\to \ell^2(\Z_-)$.
    We define
    \begin{equation}
        \oo{i} = 
        \begin{dcases}
        \seq{\chi_+}{-i-1}, \quad i \in \Z_-,\\
        \seq{\chi_+}{i}, \quad i\in \Z_+,\\    
        \end{dcases}     
    \end{equation}
    and find $P_{\Z_{\pm}}\mc J(\chi)P^*_{\Z_{\pm}}$ equals $\mc J_\pm(\chi_+)$ in all but the upper left / lower right entry, where $P_{\Z_{\pm}}$ is the coordinate projection onto the corresponding half-axis.

    Let now $\bm x, \bm y \in \ell^2(\mathbb{Z}_\pm)$ be such that $(\mc J_-(\chi_+) -\lambda) \bm y=\bm 0$ and $(\mc J_+(\chi_+) -\lambda) \bm x = \bm 0$ and define
    \begin{equation}
        \bseq{z}{i} = 
        \begin{cases}
            \beta \bseq{y}{i}, \quad i \in \mathbb{Z}_{-} ,\\
            \alpha \bseq{x}{i}, \quad i \in \mathbb{Z}_{+},
        \end{cases}
    \end{equation}
    where $\alpha, \beta \in \mathbb{R}$ are to be determined later. 
Due to the difference in boundary values \cref{eq:physicalbv} of $\mc J_\pm(\chi_+)$ and the regular bands \cref{eq:jacobibands} of $\mc J(\chi)$, tracking the mismatch in the $-1$\textsuperscript{th} and $0$\textsuperscript{th} row now yields
    \begin{equation*}
        ((\mc J(\chi)-\lambda)\bm z)^{(i)} = \begin{cases}
           \ds  \frac{v_{-1}^2}{\ell_{-1}s_{-1}}\beta\bseq{y}{-1} - \frac{v_{-1}v_0}{s_{-1}\sqrt{\ell_{-1}{\ell_{0}}}}\alpha\bseq{x}{0}, & i=-1,\\
           \ds  \frac{v_{0}^2}{\ell_{0}s_{-1}}\alpha\bseq{x}{0} - \frac{v_{-1}v_0}{s_{-1}\sqrt{\ell_{-1}{\ell_{0}}}}\beta\bseq{y}{-1}, & i=0,\\
            0, & \text{else.}
        \end{cases}
    \end{equation*}
    This system has a nontrivial solution in terms of $\alpha$ and $\beta$ if and only if
    \begin{equation*}
        \det \begin{pmatrix}
          \ds  -\frac{v_{-1}v_0}{s_{-1}\sqrt{\ell_{-1}{\ell_{0}}}}\bseq{x}{0} & \ds \frac{v_{-1}^2}{\ell_{-1}s_{-1}}\bseq{y}{-1}\\
            \ds \frac{v_{0}^2}{\ell_{0}s_{-1}}\bseq{x}{0} & \ds -\frac{v_{-1}v_0}{s_{-1}\sqrt{\ell_{-1}{\ell_{0}}}}\bseq{y}{-1}
        \end{pmatrix} = 0,
    \end{equation*}
    which is always true. Thus we can choose $\alpha, \beta\in \R$, such that $(\mc J(\chi)-\lambda)\bm z= 0$ with $\bm z \in \ell^2(\Z)$. But, by \cref{prop: big_equivalence},  $(\mc J(\chi)-\lambda)$ is invertible, which is a contradiction.
\end{proof}

\subsection{Spectral bounds for finite arrays}
Having characterised the spectra of semi-infinite Jacobi operators, we are now in a position to do so for the finite case. Intuitively, this is due to the fact that for large system sizes $M\to \infty$, the finite system approaches an infinite system with a left and right edge. The following result formalises this intuition and provides a strict spectral bound for finite systems.
\begin{theorem}\label{thm:finite_spectra}
    Define
    \begin{equation*}
        \Sigma = \bigcup_{\chi_+\in \ldzp}\sigma(\mc J_+(\chi_+))\cup \sigma(\mc J_-(\chi_+)),
    \end{equation*}
    and let $\chi_M \in D^M$ for some $M \in \N$ with the corresponding capacitance matrix eigenproblem $VC$. Then, we have $ \sigma(VC) \subset\Sigma$.    
\end{theorem}
\begin{proof}
    Let $\chi \in D^{\Z}$ and $\chi_+\in \ldzp$ be the sequences corresponding to bi-infinite and semi-infinite repetitions of the sequence $\chi_M$, respectively. For the sake of contradiction, suppose that $\mc J_N(\chi_M)-\lambda$ is singular, but $\lambda \notin \Sigma$. Let $\bm x\in \R^N$ be in the kernel of $\mc J_N(\chi_M) - \lambda$. We now define
   \begin{equation*}
        \bm z = (\dots, r_{-1}\bm x, r_0 \bm x, r_1\bm x, \dots)^\top,
    \end{equation*}
    where the sequence of real numbers $\{r_n\}_{n\in \mathbb{Z}}$ will be determined later. A direct calculation yields
   \begin{equation*}
         ((\mc J(\chi)-\lambda)\bm z)^{(i)}=\begin{cases}
           \ds  \frac{v_{N-1}^2}{\ell_{N-1}s_{N-1}}r_{k-1}\bseq{x}{N-1} - \frac{v_{N-1}v_0}{s_{N-1}\sqrt{\ell_{N-1}{\ell_{0}}}}r_k\bseq{x}{0}, & i=kN-1, k\in\Z,\\
           \ds  \frac{v_{0}^2}{\ell_{0}s_{N-1}}r_k\bseq{x}{0} - \frac{v_{N-1}v_0}{s_{N-1}\sqrt{\ell_{N-1}{\ell_{0}}}}r_{k-1}\bseq{x}{N-1}, & i=kN, k\in\Z,\\
            0, & \text{else.}
            \end{cases}
    \end{equation*}
    As before, this system has a nontrivial solution in $\{r_n\}_{n\in \mathbb{Z}}$ if and only if 
    \begin{equation*}
        \det \begin{pmatrix}
          \ds  -\frac{v_{N-1}v_0}{s_{N-1}\sqrt{\ell_{N-1}{\ell_{0}}}}\bseq{x}{0} & \ds \frac{v_{N-1}^2}{\ell_{N-1}s_{N-1}}\bseq{x}{N-1}\\
          \ds   \frac{v_{0}^2}{\ell_{0}s_{N-1}}\bseq{x}{0} & \ds -\frac{v_{N-1}v_0}{s_{N-1}\sqrt{\ell_{N-1}{\ell_{0}}}}\bseq{x}{N-1}
        \end{pmatrix} = 0,
    \end{equation*}
    which is always true. 
    
Since all off-diagonal entries of $\mc J_N(\chi_M)$ are nonzero, $\bm{x}^{(0)}$ and $\bm{x}^{(N-1)}$ do not vanish.  The two interface equations are equivalent to
\[
 r_k=\rho r_{k-1},
 \qquad
 \rho=\frac{v_{N-1}\sqrt{\ell_0}}{v_0\sqrt{\ell_{N-1}}}
       \frac{\bm x^{(N-1)}}{\bm x^{(0)}}\ne0.
\]
With $r_0=1$, we have  $r_k=\rho^k$.  In particular, one of $\{r_n\}_{n\in\mathbb{N}},\{r_n\}_{n\in\mathbb{Z}\setminus\mathbb{N}} $ is bounded.

If $|\rho|<1$, the positive-half restriction is in $\ell^2$ and is a nonzero eigenvector of $\mc J_+(\chi_+)-\lambda$.  If $|\rho|>1$, the negative-half restriction is in $\ell^2$ and is a nonzero eigenvector of $\mc J_-(\chi_+)-\lambda$.

When $|\rho|=1$, define the bounded generalised solution
\[
 \bm z_+^{(kN+i)}:=\rho^k \bm x^{(i)},
 \qquad k\ge0,\quad i=0,\ldots,N-1.
\]
It satisfies $(J_+(\chi_+)-\lambda) \bm z_+=0$ componentwise.  Let $\Pi_L$ project onto the first $LN$ coordinates and put
\[
\bm w_L:=\Pi_L \bm z_+,
 \qquad
 \bm \psi_L:=\frac{\bm w_L}{\|\bm w_L\|}.
\]
Because $|\rho|=1$,
\begin{equation}
 \|\bm w_L\|^2=L\|\bm x\|^2.
\label{eq:weyl-norm}
\end{equation}
Moreover,
\[
 (\mc J_+(\chi_+)-\lambda) \bm w_L=[\mc J_+(\chi_+),\Pi_L] \bm z_+.
\]
The commutator is supported only at the two sites adjacent to the truncation point.  Boundedness of the Jacobi coefficients and of $\bm z_+$ gives a constant $C$ independent of $L$ such that
\[
 \|(\mc J_+(\chi_+)-\lambda) \bm w_L\|\le C.
\]
Together with~\eqref{eq:weyl-norm},
\[
 \|(\bm J_+(\chi_+)-\lambda)\bm \psi_L\|
 \le\frac{C}{\sqrt L\,\|\bm x\|}\longrightarrow0.
\]
Thus $(\bm \psi_L)$ is a Weyl sequence and $\lambda\in\sigma(\mc J_+(\chi_+))\subset\Sigma$.  In every case $\lambda\in\Sigma$, contradicting the assumption $\lambda\notin\Sigma$.  This proves that
\[
 \sigma(\mc J_N(\chi_M))=\sigma(VC)\subset\Sigma.
\]

\end{proof}

This result now fully justifies the observations in \cref{fig:saxonhutner_edgemode}. In addition to always satisfying the source-sink condition, standard blocks from \cref{ex:standard_blocks} also always have $(1,0)^\top$ contained in the maximal connected set containing all sinks, preventing by \cref{cor:edgemode_characterisation} the existence of any edge modes. As a consequence, we have $\sigma(\mc J_+(\chi_+)) = \sigma(\mc J_-(\chi_-)) = \sigma(\mc J(\chi))$ for any pseudo-ergodic $\chi_+\in \ldzp$ and $\chi \in \ldz$ with $\chi_+\prec\chi$. This implies that $\Sigma = \sigma(\mc J(\chi))$ and thus also $\sigma(VC)\subset \sigma(\mc J(\chi))$ for any finite block disordered array. Conversely, for non-standard blocks allowing for the existence of edge modes, these edge modes contribute to $\Sigma$ and may thus also show up in finite systems, as can be seen in \cref{fig:saxonhutner_edgemode}(\textsc{b}). 

At this point, we would also like to point out that the converse of \cref{thm:finite_spectra} is provided by \cite[Theorem 3.9]{ammari.barandun.ea2025Universal}, ensuring that for any sequence of finite arrays converging to an infinite pseudo-ergodic one, the corresponding density of states of the finite arrays must converge to the infinite one.
By \cref{thm:saxonhutner}, we know that the bands of the constituent blocks where one or more of their block propagation matrices fail to be hyperbolic, must be filled out for any infinite system with a pseudo-ergodic block sequence. The convergence of the finite systems to the infinite systems in terms of the density of states then ensures that the bands get filled out, also for the finite systems. Furthermore, this convergence also ensures that, while edge modes in the bandgap may be possible for the finite system, they constitute an increasingly small fraction of the overall eigenmodes, as the system size increases.

\begin{remark}
Finally, we would like to remark how \cref{thm:coburn,thm:finite_spectra} are connected to the uniform hyperbolicity of the block propagation matrix cocycle. For \cref{thm:coburn} we take $\chi_+\in\ldzp$ pseudo-ergodic and $\chi$ the bi-infinite doubling of $\chi_+$, as in the proof of \cref{thm:coburn}. Having both left and right edge modes at some $\lambda$ in the gap would then imply $u(0) = s(0)$ of the block propagation matrix cocycle associated with $\chi$, a contradiction due to its uniform hyperbolicity.

\cref{thm:finite_spectra} can be obtained after realising that, for some $\chi_M \in D^M$, we have
\[
    \lambda \in \sigma(\mc J_N(\chi_M)) \iff \mc P_{tot}^\lambda \begin{pmatrix}
        1\\0 
    \end{pmatrix} = \xi \begin{pmatrix}
        1\\0
    \end{pmatrix},
\]
where $P_{tot}^\lambda = \prod_{j=1}^M\mc P_{\seq{\chi_M}{j}}^\lambda$. Because this implies that $(1,0)^\top$ is an eigenvector of $\mc P_{tot}^\lambda$, we can (depending on whether $\abs{\xi}$ is larger or smaller than one) extend $(1,0)^\top$ into an eigenvector of either the negative or positive semi-infinite system.
\end{remark}

\subsection*{Acknowledgements}
   The work was supported in part by the Swiss National Science Foundation grant number 200021-236472.  
   
\subsection*{Data Availability Statement} The authors confirm that the current paper does not have associated data.
\subsection*{Competing Interests Statement} The authors declare no competing interests.


\appendix

\section{Dynamically defined cocycles}\label{sec:dyndefined}
In this paper, we have considered $\operatorname{SL}(2,\R)$-cocycles as maps $A\in \ell^\infty(\Z,\operatorname{SL}(2,\R))$ that are determined by a block sequence $\chi\in \ldz$ or a resonator sequence $\alpha\in \lrz$. Since in this case the cocycle is essentially modelled as a sequence of $\operatorname{SL}(2,\R)$-matrices, we call it \emph{sequentially defined}.

However, there is an alternative but closely related formalism of \emph{dynamically defined cocycles}. The key difference now is to consider the cocycle map $A:\Omega\to \operatorname{SL}(2,\R)$ to be a map on $\Omega$, the space of all allowable sequences. This together with a \emph{shift operator} $S:\Omega\to \Omega$ defines the following dynamical system:
\begin{align*}
    (S,A):\Omega\times\R^2 &\to \Omega\times\R^2\\
    (\chi, \bm v) &\mapsto (S(\chi), A(\chi)\bm v).
\end{align*}

In this formalism, uniform hyperbolicity is defined as follows.
\begin{definition}[Uniform hyperbolicity for dynamically defined cocycles]\label{def:DynamicUH}
We say that a dynamically defined cocycle $(S, \Omega, A)$
is \emph{uniformly hyperbolic} (UH) if
there exist two maps $s,u:\Omega\to \rp$ such that
\begin{enumerate}[label=(UH\arabic*)]
  \item \textbf{$A$-invariance.}  
        For every $\chi\in \Omega$,
        \[
          A(\chi)u(\chi) = u(S(\chi))
          \quad\text{and}\quad
           A(\chi)s(\chi) = s(S(\chi));
        \]
  \item \textbf{Uniform growth.}  
        There exist $C>0$ and $\eta>1$ such that
        \[
          \norm{A_{-n}(\chi)\vec{u}}, \norm{A_n(\chi)\vec{s}} \leq C\eta^{-n}
        \]
        for all $\chi\in\Omega$, $n\in \N$ and all unit vectors
        $\vec{u}\in u(\chi)$ and
        $\vec{s}\in s(\chi)$.
\end{enumerate}
Here, analogously to the sequential case, $A_n(\chi) = A(S^{n-1}(\chi))\cdots A(\chi)$.
\end{definition}

To illustrate the difference between these two formalisms, we consider the cocycle of block propagation matrices. In both cases, the key ingredient in their construction is the map
\begin{equation}\label{eq:blockpropmap}
    \begin{aligned}
        \Psi:\{1,\dots,D\}&\to \operatorname{SL}(2,\R)\\ 
        d &\mapsto \mc P^\lambda_d,
    \end{aligned}
\end{equation}
taking a block index $d$ to the propagation matrix of the corresponding block $\mc P^\lambda_d$, as constructed in \cref{def:blockprop}. 

In \cref{def:blockprop}, $\Psi$ was implicitly used, for any $\chi\in \ldz$, to sequentially define the cocycle
\begin{align*}
    \mc P^\lambda:\Z&\to \operatorname{SL}(2,\R)\\
    j &\mapsto \Psi(\seq{\chi}{j}) = \mc P^\lambda_{\seq{\chi}{j}}.
\end{align*}

Alternatively, consider $\Omega\subset \ldz$ so that the shift map $S:\ldz \to \ldz$ restricted to $S:\Omega\to \Omega$ is a well-defined homeomorphism. Using $\Psi$, we can then dynamically define the cocycle as 
\begin{align*}
    \mc P^\lambda:\Omega&\to \operatorname{SL}(2,\R)\\
    \chi &\mapsto \Psi(\seq{\chi}{0}) = \mc P^\lambda_{\seq{\chi}{0}}.
\end{align*}

For these two formalisms, the concepts of uniform hyperbolicity are related as follows \cite[Proposition 1]{zhang2020Uniform}.
\begin{proposition} 
    Let $\lambda\geq0$ and let $\chi_0\in \ldz$.  Let 
    $\Omega_{\chi_0} :=\overline{\operatorname{Orb}(\chi_0)}$. Define
\[
 \mc{P}^\lambda(\omega):= \mc{P}_{\omega(0)}^\lambda,
 \qquad\omega\in\Omega_{\chi_0},
\]
and the sequential cocycle
\[
\mc{P}_{\mathrm{seq}}^\lambda(j)
 :=\mc{P}^\lambda(S^j\chi_0)= \mc{P}_{\chi_0(j)}^\lambda.
\]
Then, $\mc{P}_{\mathrm{seq}}^\lambda$ is uniformly hyperbolic if and only if the dynamically defined cocycle $(S,\Omega_{\chi_0},\mc{P}^\lambda)$ is uniformly hyperbolic  in the sense of \cref{def:DynamicUH}.


\end{proposition}


An immediate consequence is that if $\chi_0$ is pseudo-ergodic, then
\[
 \Omega_{\chi_0} =\ldz.
\]
Consequently, uniform hyperbolicity of the sequential cocycle along $\chi_0$ is equivalent to uniform hyperbolicity of the dynamically defined cocycle over the full shift, or in other words, to the fact that $\{ \mc{P}_1^\lambda, \ldots, \mc{P}^\lambda_D\}$ admits a strictly invariant multicone; see \cite[Theorem 2.2]{avila.bochi.ea2010Uniformly}.

\printbibliography

\end{document}